\documentclass[conference]{IEEEtran}
\IEEEoverridecommandlockouts
\usepackage[dvipsnames]{xcolor}
\usepackage{amsmath,amsthm,amsfonts}
\usepackage[noEnd=true, rightComments=false]{algpseudocodex}
\usepackage{graphicx,subcaption}
\usepackage{xspace}
\usepackage{textcomp}
\usepackage{algorithm}
\usepackage{enumitem}
\usepackage{dsfont}
\usepackage{diagbox}
\usepackage{array}
\usepackage{flushend}
\usepackage [
    n, 
    advantage,
    operators,
    adversary,
    landau,
    probability,
    notions,
    ff, 
    mm,
    primitives,
    events,
    complexity,
    oracles,
    asymptotics,
    keys
]{cryptocode}
\usepackage{todonotes}
\usepackage{mathrsfs}
\usepackage{tikz}
\usepackage{url}
\usepackage{multirow}
\usepackage{caption}
\usepackage{balance}
\usetikzlibrary{arrows}
\usepackage{hyperref}
\usepackage{cleveref}
\usepackage{cite}
\usepackage[most]{tcolorbox}

\graphicspath{{figure/}}

\def\BibTeX{{\rm B\kern-.05em{\sc i\kern-.025em b}\kern-.08em
    T\kern-.1667em\lower.7ex\hbox{E}\kern-.125emX}}

\def\eg{\textit{e.g.}\xspace}
\def\aka{\textit{a.k.a.}\xspace}
\def\ie{\textit{i.e.}\xspace}
\def\etc{\textit{etc.}\xspace}
\def\etal{\textit{et~al.}\xspace}

\newcommand{\sys}{\textsf{Appraisal}}
\newcommand{\true}{\mathsf{true}}

\newtheorem{theorem}{Theorem}[section]

\theoremstyle{definition}
\newtheorem{definition}{Definition}[section]
\theoremstyle{remark}

\begin{document}

\title{Privacy-Preserving Screening for Record Linkage \\
\thanks{Fan Zhang and Peng Chen are corresponding authors.}
}

\author{\IEEEauthorblockN{
Chenyu Huang\IEEEauthorrefmark{1}\IEEEauthorrefmark{2},
Fan Zhang\IEEEauthorrefmark{1}\IEEEauthorrefmark{2}, 
Huangxun Chen\IEEEauthorrefmark{3}, 
Yongjun Zhao\IEEEauthorrefmark{4}, 
Huaming Rao\IEEEauthorrefmark{1},
Peng Chen\IEEEauthorrefmark{1}, Danqing Huang\IEEEauthorrefmark{1}
}
\IEEEauthorblockA{
\IEEEauthorrefmark{1}Tencent Inc,
\IEEEauthorrefmark{3}Hong Kong University of Science and Technology (GZ), 
\IEEEauthorrefmark{4}Independent Researcher
\IEEEauthorrefmark{2}Co-primary Authors
}

\IEEEauthorblockA{
\{chenyuhuang, zxfanzhang\}@tencent.com, huangxunchen@hkust-gz.edu.cn, \\foreverjun.zhao@gmail.com, 
\{huamingrao, pengchen, daisyqhuang\}@tencent.com
}
}








\maketitle
\begin{abstract}

In an era dominated by big data and machine learning, establishing valuable data collaboration has never been more critical. 
However, such collaborations must operate under regulatory and legal constraints. 
Two-party Privacy-Preserving Record Linkage (PPRL) emerges to assess the potential collaboration value and also ensure the privacy and security of the involved data. 
Nevertheless, the substantial computational and communication overheads associated with PPRL hinder its practical adoption in data markets with numerous potential collaborators. 
Therefore, we present the \emph{Screening-then-Linkage} framework, which incorporates a lightweight \emph{Screening} phase prior to the resource-intensive PPRL phase, \ie, PPRS, to mitigate the scalability issue of PPRL.
We propose a circuit-PSI-based system, named \sys~to realize a secure, effective, and efficient PPRS. 
To reconcile the approximate matching and/or schema-aware setting required in PPRS with the limitations of the circuit-PSI supporting only symmetric functions, we propose a more communication-efficient secure permutation, \ie, Oblivious Attribute/Feature Alignment protocol tailored for PPRS. 
This protocol supports a broader range of comparison functions and significantly improves efficiency, \ie, reducing communication costs by a factor of 14 compared to the conventional protocol. 
Our rigorous analysis and comprehensive empirical evaluations demonstrate the security, effectiveness, and efficiency of \sys. 
\sys~can accommodate up to $850\times$ more records than the SOTA PPRS system, SFour, within the same constraints.
Moreover, it is $165 \times$ faster than SOTA PPRL, indicating the \emph{Screening-then-Linkage} framework substantially decreases the computation time required to identify the most valuable collaborators from a large pool of candidates.

\end{abstract}


\section{Introduction}
\label{sec:introduction}

Empowered by the progress of machine learning techniques,
organizations endeavor to acquire more valuable data to enhance their analytical and predictive models~\cite{asudeh2022towards, wang2018deeper, hardy2017private, nock2021impact} in sectors like finance~\cite{Fincom}, healthcare~\cite{kuehni2012cohort, datavantmatch}~\etc
These trends stimulate the growth of data federation markets, where data owners can discover potential collaborators or engage in data trading. 

An indispensable operation in data markets is to assess the~\emph{collaboration value} across multiple datasets or parties.

\noindent$\bullet$ 
\emph{Example of Advertisement:}
Amazon Publisher Service~\cite{awspublisherservice} provides a platform that enables advertisers and publishers to access comprehensive data from both sides which helps them optimize advertising planning and investment.
Advertisers need to select the most suitable publisher for an advertiser, \ie, one that maximizes profit, and vice versa.
By calculating the match rate (collaboration value), the percentage of common customers, via record linkage service~\cite{awsrecordlinkage}, both parties can make more informed decisions, thereby enhancing the effectiveness of their advertising campaigns.

\noindent$\bullet$ 
\emph{Example of Healthcare:} 
Datavant~\cite{datavant}, the leading health data marketplace, has records from over 70K hospitals and clinics.
Its platform facilitates hospitals in finding appropriate datasets via its record linkage service~\cite{datavantrecordlinkage}.
A particularly challenging area is rare diseases, where it’s difficult to find patients.
Pharmaceutical companies invest significant time linking their patients with rare diseases to various datasets~\cite{datavantraredisease, conf/asiaccs/WuVKR23}.
Unfortunately, their efforts often yield no results due to the small number of patients.
Therefore, they are eager to find an efficient way to determine whether the datasets contain the patients they are interested in.


In general, the collaboration value may vary on a case-by-case basis and is highly related to the statistics of the linked records, such as the quantity/ratio of linked/unlinked data. 
Meanwhile, to foster a prosperous data market under the data regulations and privacy laws, such as General Data Protection Regulation and China's Personal Information Protection Law, we need a secure, effective, and efficient way to estimate the collaboration value for multiple involved parties, \ie, plaintext-based record linkage is not allowed. 
Specifically, it must ensure that no information is disclosed except for the estimation result (\emph{security}); it should provide accurate estimation (\emph{effectiveness}); and it should run efficiently to swiftly identify the most valuable dataset among massive candidates (\emph{efficiency}).

Numerous two-party \textbf{P}rivacy-\textbf{P}reserving \textbf{R}ecord \textbf{L}inkage (PPRL) methods~\cite{wong2013privacy, essex2019secure,phua2010resilient, gkoulalas2021modern, khurram2020sfour, adir2022privacy} have been proposed
. 
Among these, Multi-party Computation (MPC)-based PPRL~\cite{wong2013privacy, essex2019secure,khurram2020sfour, wei2023cryptographically} ensure security while maintaining high accuracy of linkage. However, they suffer from high computational costs. For example, the performance of SFour~\cite{khurram2020sfour} is bottlenecked by 
cryptographic sorting and comparison;
Wei~\etal~\cite{wei2023cryptographically} involves $\mathcal{O}(n\log n)$ secure comparisons.
As a result, the best solution can only link 40K records per hour~\cite{wei2023cryptographically}.
The efficiency issue becomes even more severe when dealing with a large number of potential candidates in the market.

Instead of further optimizing the efficiency of PPRL itself, 
we observe that PPRL may be an overkill in the collaborative candidate search phase. 
Before pinpointing each linked record, an involved party might be more interested in the aggregated statistics of the linked records (collaboration value), such as the number of linked records with each potential candidate:  
if the number is too small, it is a waste of time to compute the exact linking record pairs. 
With this observation, we propose a new framework for data collaboration called~\emph{Screening-then-Linkage}~(Fig.~\ref{fig:Screening_then_Linkage}), and develop a lightweight \textbf{P}rivacy-\textbf{P}reserving \textbf{R}ecord \textbf{S}creening (PPRS) system, \sys. Our main idea is to employ a lightweight protocol to securely, effectively, and efficiently estimate \textbf{the aggregated statistics of linked records} with all potential candidates, so as to select the most suitable targets as the input to the subsequent PPRL procedure. 
\begin{figure}[t]
\centering
\includegraphics[width=\linewidth]{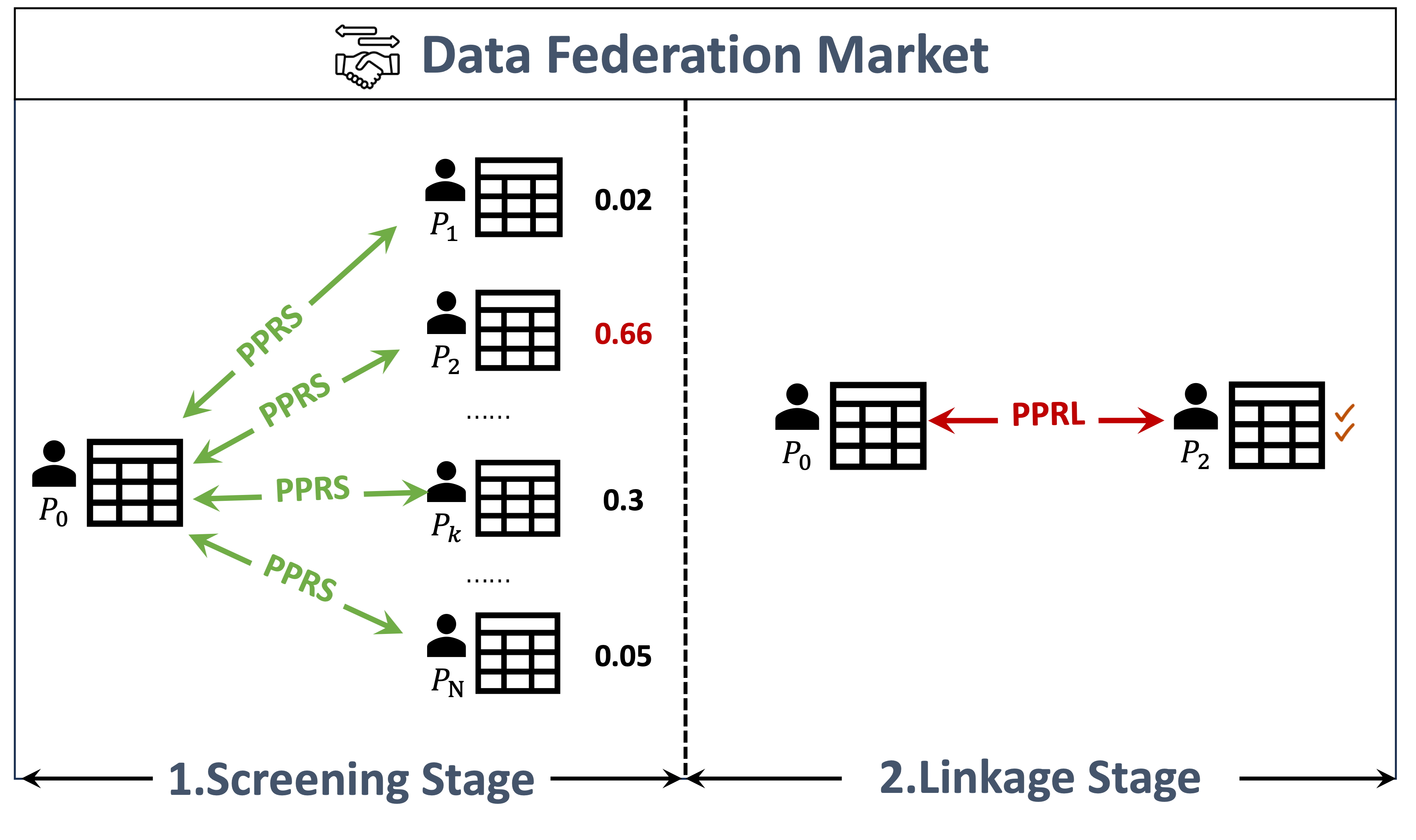}
\caption{\emph{Screening-then-Linkage} Framework.} 
\label{fig:Screening_then_Linkage}
\end{figure}

Notice that the relaxed goal of PPRS allows us 
to achieve efficiency and effectiveness simultaneously, in contrast to using PPRL directly. 
Instead of naively comparing all parties between databases, most PPRL works rely on \emph{blocking} strategy that assigns records to different bins, and only compares the subset of the pairs in the corresponding bins. 
In PPRL, opting for a smaller bin size tends to compromise effectiveness in favor of efficiency, given the task’s fine-grained nature and sensitivity to the bin size. 
Empirically, bin size is set that in most works that requires at least $\mathcal{O}(n\log n)$ costly secure comparisons~\cite{wei2023cryptographically, he2017composing}.
For PPRS, we can reduce the record screening problem to the set intersection problem: considering a record with multiple attributes, linking a record is determined by whether most of its attribute values have appeared in the opposite databases, which only involves $\mathcal{O}(n)$ secure comparisons, which is more lightweight than PPRL. Set intersection does not involve bins as in PPRL and is therefore unaffected by the impact of bin size.


However, realizing a secure protocol based on the idea is not trivial.
The most relevant technique is private set intersection (PSI), a protocol that can reveal the intersection result without leaking the elements outside the intersection.
Adir~\etal~\cite{adir2022privacy} implements the PSI-based PPRL, which takes all the attributes from one record as input, and two parties can determine whether the record is linked or not.
However, it can only support schema-agnostic settings with exact matching, which restricts the usage scope.
A straightforward solution is to input all the attribute values to PSI separately to support schema-aware setting, or input the signature bins derived from Locality Sensitive Hashing (LSH) to support approximate matching~\cite{adir2022privacy}.
However, both parties will know the intersected attribute value distribution even the records are not linked, which violates the privacy requirements of PPRL.
To preserve the privacy of all the attributes, we choose 
circuit-PSI protocol~\cite{pinkas2019efficient, chandran2022circuit, rindal2021vole}, which is an advanced variant of PSI to enable secure computation over the intersection.
Unlike PSI, circuit-PSI outputs the boolean shares of the intersection result, preserving the privacy of the attribute value.
However, it restricts that \emph{only symmetric functions} can be applied to the intersection items, namely the output is not related to the order of elements in the intersection.
Therefore, it cannot naturally support two basic \emph{asymmetric permutation} functions: (i) replicating elements and (ii) reordering elements. 
But these functions are essential to support approximate matching, schema-aware setting, and complex scoring functions for collaboration value in PPRS system.
To address this, we propose an Oblivious Feature Alignment (OFA) protocol to support these by enabling oblivious data permutation that builds upon secure permutation techniques~\cite{mohassel2013hide}.
Moreover, we incorporate various optimizations in OFA to make it practical in PPRS, which greatly improves efficiency.

Our contributions can be summarized as follows:

$\bullet$ We identify the discrepancy between the needs of the data market and the capabilities of existing PPRL methods and thus introduce the \emph{Screening-then-Linkage} framework and the PPRS system, \sys, specifically to address the issue.

$\bullet$ \sys~incorporates the circuit-PSI and a computationally efficient OFA protocol, ensuring comprehensive support for various operations vital to PPRS, including both exact and approximate matching of attributes.

$\bullet$ We validate the security, effectiveness, and efficiency of \sys\ through rigorous analysis and extensive evaluation. 
Our system is $165 \times$ faster than state-of-the-art (SOTA) PPRL~\cite{wei2023cryptographically} on million-level records, and can accommodate up to $850$ times more records within the same time frame than the SOTA PPRS system, SFour~\cite{khurram2020sfour}.

\section{Problem Definition}
\label{sec:problem}
\subsection{Notations}
For a finite set $S$, $s \leftarrow S$ denotes choosing a uniformly random element from $S$ and 
$|S|$ denotes the number of elements in $S$. 
We use bold lowercase for vectors (\eg, $\mathbf{a}$) and uppercase for matrices (\eg, $\mathbf{M}$).
$\mathbf{a}[i]$ (resp.~$\mathbf{M}[i]$) denotes the $i$-th entry (row), and $\mathbf{M}[i, j]$ is the value at the row $i$, column~$j$.


\subsection{System Model}


We define a party denoted as $P_0$ that seeks to identify data providers of superior quality among $N$ other parties $(P_1, P_2, ..., P_N)$ in the data market.
Each party possesses a database denoted as $\mathbf{V}^i$ from domain $\mathbf{V}$, comprising a collection of $n_i$ records from domain $R$. 
These parties agree on a matching function $\mathsf{m}:R\times R \rightarrow \{0,1\}$.
The goal of the PPRS is to get $\mathsf{f}(\{v | v \in \mathbf{V}^0, v' \in \mathbf{V}^i, \mathsf{m}(v, v')=1\})$ between $P_0$ and $P_i$ where $\mathsf{f}(\cdot)$ is a secure score function on $P_0$'s linked records $\{v\}$.
Without loss of generality, we assume $n_0 = n_1 = ... = n$, nevertheless, our solution accommodates scenarios where each party holds a varying number of records.

Two modes are supported~\cite{papadakis2023analysis}: (i) \emph{Schema-aware} knows the schema information and focuses on the most relevant attributes. Without loss of generality, we assume all parties have the same attributes here. (ii) \emph{Schema-agnostic} does not take any schema information into account and instead uses the whole attribute values as matching input. 

Two matching types are supported~\cite{papadakis2023benchmarking}: (i) \emph{Exact matching} determines the match of attribute value iff they are exactly the same, unless the value is missing; (ii) \emph{Approximate matching} estimates a matching probability given two attribute values, thus can tolerate the typos and different syntax representations. 

\subsection{PPRL's Trade-off b/w Efficiency and Effectiveness/Security}
\label{sec:problem_pprl}
PPRL aims at finding the exact link between the records, \eg, all linked pair $(v^0, v^1)$ where $v^0\in \mathbf{V}^0, v^1 \in \mathbf{V^1}$ between the $P_0$ and $P_1$.
However, the huge time cost hinders them far from practical, \ie, SOTA~\cite{wei2023cryptographically} takes hours to link 40K records from two parties. 
The workflow contains two stages:
(i) Blocking: both parties assign the records to multiple bins via blocking.
(ii) Matching: execute the attribute-level matching between the pairwise bins. 
The number of secure comparisons within each bin is $\mathcal{O}((n/\ell)^2)$ where $\ell$ is the number of bins, resulting in $\mathcal{O}(n^2/\ell)$ comparisons of the whole protocol.
With larger $\ell$, the number of secure comparisons will decrease a lot.
However, this may cause a high false negative rate due to the probability of candidate pairs of record, one from $P_0$ and the other from $P_1$, assigning to the corresponding bin decrease.
To balance the effectiveness and efficiency, $\ell$ is set to  $\mathcal{O}(n \log n)$ empirically~\cite{he2017composing,wei2023cryptographically},  resulting a time complexity $\mathcal{O}(n \log n)$ of the whole protocol.
This is only an ideal situation without considering the security.
The number of records within the bin is sensitive information, which can leak the data distribution, for example, the block bin based on ``gender'' attribute can leak gender ratio.
Thus, padding around $\eta$ dummy elements in each bin is required to preserve the privacy of bin size.
The SOTA~\cite{wei2023cryptographically} achieves $\eta=n/2\ell$.
Therefore, there are at least $\ell (n/\ell+\eta)^2\approx 2n\log n$ secure comparisons.


Consider a data federation market with a party, $P_0$ and 10 candidates, each has 1M records. It takes around one month (10$\times$70 hours per PPRL) for $P_0$ to identify the most valuable partners.
We observe that it is not necessary to find the \emph{exact linkage} between $P_0$'s database and all candidates. Instead, we can first quickly estimate the \emph{collaboration value}, \eg, number of linked records, so that we can greatly shrink the candidates' size. Then, $P_0$ only needs to execute PPRL with high-valued candidates, thereby reducing computing costs.



\subsection{Rationale of \emph{Screening-then-Linkage} Framework}
\label{sec:problem_framework}

To overcome the above issue, we develop a \emph{Screening-then-Linkage} framework
(Algorithm~\ref{alg:screening_then_linkage}):
(i) $P_0$ initiates the PPRS protocol with other parties and obtains the collaboration value $c$ (Line \ref{alg:screening_then_linkage:L2}--\ref{alg:screening_then_linkage:L3});
(ii) $c$ is compared against a predefined threshold, and those that exceed the threshold are added to set $U$ (Line \ref{alg:screening_then_linkage:L4}--\ref{alg:screening_then_linkage:L6});
(iii) $P_0$ gets the exact linkage by executing PPRL with the parties in $U$ (Line \ref{alg:screening_then_linkage:L8}--\ref{alg:screening_then_linkage:L10}).

\renewcommand{\algorithmicrequire}{\textbf{Input:}}
\renewcommand{\algorithmicensure}{\textbf{Output:}}
\begin{algorithm}[t]
	\caption{\emph{Screening-then-Linkage} Framework}           
	\label{alg:screening_then_linkage}               
	\begin{algorithmic}[1]                
	\Require Data $\mathbf{V}^i$ of party $P_i$ ($i \in [0, N]$);
        \Ensure Set of linked records $\tilde{\mathbf{V}}$
        \State $U = \emptyset$
        \For{$i = 1$ to $N$} \label{alg:screening_then_linkage:L2}
        \State $P_0$ recoveries $c$ from the output of $\mathsf{PPRS}(\mathbf{V}^0, \mathbf{V}^i)$  \label{alg:screening_then_linkage:L3}
        \If {$c > threshold$}  \label{alg:screening_then_linkage:L4}
        \State $U = U \cup \{ i \}$ 
        \EndIf  \label{alg:screening_then_linkage:L6}
        \EndFor
        \For{$i \in U$}  \label{alg:screening_then_linkage:L8}
        \State $\tilde{\mathbf{V}}^{i} = \mathsf{PPRL}(\mathbf{V}^0, \mathbf{V}^i)$ \label{alg:screening_then_linkage:L10}
        \EndFor  
        \State \Return $\tilde{\mathbf{V}} \gets \{\tilde{\mathbf{V}}^{i} | i\in U\}$
	\end{algorithmic}
\end{algorithm}

The proposed framework offers several advantages. First,  executing the lightweight PPRS protocol to get the collaboration value before pinpointing exact linked records 
reduces the overall execution time.
Secondly, PPRS ensures that no information other than the final collaboration value is leaked to $P_0$, thereby protecting the data privacy of the involved parties. 
Additionally, if we allow the disclosure of whether the records of $P_0$ are linked, the subsequent PPRL can operate on these linked records instead of all the records.

\subsection{Screening for Record Linkage}
\label{sec:screeing4rl}
To implement a lightweight screening, reducing the number of comparisons is in urgent need.
Thus, our core idea is to reduce the screening problem to the set intersection problem, which can be efficiently solved with only $n$ secure comparisons, such as merge sort.
Considering the most simple case - the exact matching in a schema-agnostic setting.
This process can be straightforwardly executed by merging all values from different attributes into a single, extended string. 
Then, it is reduced to the set intersection problem: finding the records in common from two sets $\mathbf{V}^0$ and $\mathbf{V}^1$


\subsubsection{Extend to Approximate Matching} 
\label{sec:s4rl_approx}
There are two main types of approximate matching: (i) fuzzy matching for string attributes, and (ii) distance matching for numeric attributes. Specifically, we employ locality-sensitive hashing 
(LSH)~\cite{indyk1998approximate}, a highly effective approach. It identifies the nearest neighbor in high-dimensional space and outputs identical hash values for similar inputs. 
We adopt MinHash LSH for string fuzzy matching and cross-polytope LSH~\cite{andoni2015practical} for numeric distance matching for their efficiency and effectiveness.
Without loss of generality, let’s consider MinHash LSH as an example. It approximates the Jaccard coefficient between sets derived from candidate record pairs. A string is divided into a set, $E$, of all overlapping substrings of length $q$ (q-grams). 
Then $B\cdot R$ different MinHash functions are adopted, each will be applied to $E$ and output hash value $\mathsf{min}_{e\in E} h_k(e)$ where $h_k(\cdot)$ is the $k$-th hash function.
The $B\cdot R$ values are grouped into $B$ bands of $R$ hash values, and the concatenation of the $R$  hash values is the band's signature.
A candidate attribute pair is considered a match if at least one of the band signatures matches. Cross-polytope LSH functions similarly for numeric values, which also hashes attribute values into $B$ bands.
To incorporate LSH with the set intersection, each attribute will be expanded to $B$ band signatures via MinHash.
An attribute is considered to be matched if one of the $B$ signatures, for example, $j$-th signature, exists in the other party's set of $j$-th signatures.

\subsubsection{Extend to Schema-aware} 
\label{sec:s4rl_schema_aware}
In this mode, the parties will align the schema between them. The straightforward approach is to run a set intersection at the attribute level, \ie, the value of the $j$-th attributes is matched iff it exists in the other party's set of $j$-th attributes.
Then a record is considered to be linked if most of its attribute values is matched (or via a score function).
However, this may not be suitable for attributes with low discriminative power, \ie, those having limited attribute space.
Take the ``gender'' attribute as an example: the values ``male'' and ``female'' are common across both $\mathbf{V}^0$ and $\mathbf{V}^1$, potentially leading to a high false positive rate for screening. 
To mitigate this, we can derive new attributes, \eg ``name\_gender'', by concatenating low and high discriminative attributes. 
This strategy can also be used when some attributes are not available in other parties, \ie, the unavailable attribute can be concatenated as one derived attribute.
Now, the linkage criteria of records is determined from these derived attributes.
The specific scheme for deriving new attributes can be established through empirical testing.

To prove the effectiveness of this idea, we conduct the evaluation in Section~\ref{sec:eva_on_sys} with 8 (4 pairs) datasets. The balanced accuracy rate is all larger than 90\%, which is competitive.

\section{Preliminaries}

\subsection{Threat Model}
We consider a static, semi-honest probabilistic polynomial time (PPT) adversary~~\cite{khurram2020sfour, adir2022privacy, wei2023cryptographically}, which follows the protocol without deviation but tries to learn more beyond the outputs.
Formally, let $\varPi$ be a two-party protocol computing a deterministic functionality $\mathcal{F}: \{0, 1\}^* \times \{0, 1\}^* \mapsto \{0, 1\}^* \times \{0, 1\}^*$, 
The output of $\mathcal{F}$ is a pair,
$\mathcal{F}_0(x_0, x_1)$ to $P_0$ 
and $\mathcal{F}_1(x_0, x_1)$ to $P_1$.
For $i \in \{0, 1\}$, the view of $P_i$ during an execution of $\varPi$ on $(x_0, x_1)$ is denoted by $\mathcal{V}^{\varPi}_i(x_0, x_1)$.
\begin{definition}
	For a function $F$, $\varPi$ privately computes $\mathcal{F}$ if there are PPT algorithms $\mathcal{S}_0$ and $\mathcal{S}_1$ such that 
$\mathcal{S}_i(x_0, \mathcal{F}_1(x_0, x_1)) \approx_c \mathcal{V}^{\varPi}_i(x_0, x_1)$
$\forall x_0, x_1 \in \{0, 1\}^*, i \in \{0, 1\}$,
	where $\approx_c$ denotes computational indistinguishability.
\label{def:2pc}
\end{definition}

We prove the security of our protocols in the~\emph{hybrid model} with composition theorem as follows.
\begin{theorem}[Composition~\cite{cu/Goldreich2004}] 
\label{theorem:composition}
Let $\mathcal{F}_1, \ldots, \mathcal{F}_m$ be ideal $2$-party functionalities, and let $\varPi$ be a $2$-party protocol in the computational $(\mathcal{F}_1, \ldots, \mathcal{F}_m)$-hybrid model, where at most one ideal evaluation call occurs per round.
Let $f_1, f_2, \ldots, f_m$ be the $2$-party protocols that evaluate $\mathcal{F}_i$ in the computational setting.
For any PPT passive real-world adversary $\mathcal{A}$ and any PPT environment machine $\mathcal{Z}$,
there exists a PPT passive adversary $\mathcal{S}$ in the $(\mathcal{F}_1, \ldots, \mathcal{F}_m)$-hybrid model such that
$	\text{Sim}_{\varPi, \mathcal{S}, \mathcal{Z}}^{\mathcal{F}_1, \mathcal{F}_2, \ldots, \mathcal{F}_m}\approx_c\text{Exec}_{\varPi^{(f_1, f_2, \ldots, f_m)}, \mathcal{A}, \mathcal{Z}}$.
\end{theorem}


\subsection{Cryptographic Building Blocks}

\subsubsection{Obilivous Transfer (OT)}
1-out-of-n OT~\cite{rabin1981exchange} is a secure two-party computation protocol between a sender inputting $n$ messages and a receiver inputting a selection index $i\in[0,n-1]$. 
The receiver only gets the $i$-th message while the sender is oblivious to which message the receiver gets.  
We denote its ideal functionality as $\mathcal{F}_{\mbox{OT}}$. 

\subsubsection{Secret Sharing Schemes}


\hfill

\textbf{(2,2)-arithmetic additive secret sharing~\cite{demmler2015aby}.}
For an $l$-bit value $x$, it is represented as $x=\langle x \rangle^A_0+\langle x \rangle^A_1$, where $\langle x \rangle^A_b \in \mathbb{Z}_{2^l}$ is held by $P_b$ for $b\in\{0,1\}$. 
For floating-point numbers $\widetilde{x}\in \mathbb{R}$, they are lifted using a specified precision $f$ to encode them as fixed-point numbers $x=\lceil \widetilde{x}2^f \rceil \in [-2^{l-1},2^{l-1}]$. 
This scheme supports the following operations:

\noindent $\bullet$ Addition. $\langle z \rangle^A=\langle x \rangle^A + \langle y \rangle^A$ is run by either parity locally.

\noindent $\bullet$ Multiplication. $\langle z \rangle^A= \mathcal{F}_{\mbox{MUL}}(\langle x \rangle^A, \langle y \rangle^A)$, which can be computed using Beaver's triples~\cite{beaver1992efficient}.

\textbf{2-out-of-2 XOR-based boolean sharing~\cite{demmler2015aby}.}
A boolean value $x$ can be represented as $x=\langle x \rangle^B_0\oplus \langle x \rangle^B_1$, where $\langle x \rangle^B_b \in \mathbb{Z}_2$ is held by party $P_b$ for $b \in \{0,1\}$. 
The boolean sharing scheme supports the following operations:

\noindent $\bullet$  XOR. $\langle z \rangle^B = \langle x \rangle^B \oplus \langle y \rangle^B$ is run by either parity locally. 

\noindent $\bullet$ NOT. $\langle y \rangle^B = \neg \langle x \rangle^B$  is only run by one parity locally. 
    
\noindent $\bullet$ AND. $\langle z \rangle^B = \mathcal{F}_{\mbox{AND}}(\langle x \rangle^B, \langle y \rangle^B)$ is achieved by a secure 2PC protocol using pre-computed boolean triples.~\cite{demmler2015aby}

\noindent $\bullet$ OR: $\langle z \rangle^B = \mathcal{F}_{\mbox{OR}}(\langle x \rangle^B, \langle y \rangle^B)$ 
is achieved by $x \vee y = x \land y\oplus(\neg x \oplus \neg y)$ using $\mathcal{F}_{\mbox{AND}}$.

\noindent $\bullet$ Multiplexer. $\langle z \rangle^A = \mathcal{F}_{\mbox{MUX}}(\langle x \rangle^A, \langle y \rangle^B)$. If $y$ is $\true$ then $z=x$ else $z=0$. It can be realized via two calls to $\mathcal{F}_{\mbox{OT}}$.~\cite{rathee2020cryptflow2}
    
\noindent $\bullet$ Boolean share to arithmetic share. $\langle y \rangle^A = \mathcal{F}_{\mbox{B2A}}(\langle x \rangle^B)$. It can implemented via one call to $\mathcal{F}_{\mbox{OT}}$~\cite{rathee2020cryptflow2}.

\noindent $\bullet$ Most significant bit. $\langle y \rangle^B = \mathcal{F}_{\mbox{MSB}}(\langle x \rangle^A)$. If $x \leq 0$ then $y$ is $\true$/1 else $y$ is false/0. This can be realized via $\mathcal{F}_{\mbox{OT}}$~\cite{rathee2020cryptflow2}.

\subsubsection{Circuit-PSI}
\label{sec:preliminaries_circuit_psi}

Private Set Intersection (PSI) allows two parties to securely compute the intersection $X\cap Y$ of their respective sets $X$ and $Y$ while keeping items non-intersection items private.
Circuit-PSI~\cite{pinkas2019efficient, chandran2022circuit, rindal2021vole}, the functionality of which denotes as $\mathcal{F}_{\mbox{CPSI}}$, takes a step further by allowing computations on the intersection, denoted as $f(X\cap Y)$, without leaking any other information. 
Specifically, 
both parties use the same set of hash functions $\{h_1, h_2, ..., h_a\}$, and  
the receiver maps $X$ to a table containing $(1+\epsilon) |X|$ ($\epsilon > 0$, used for avoiding collisions) bins using \emph{Cuckoo hashing}, while the sender maps $Y$ using \emph{simple hashing}. 
For \emph{Cuckoo hashing} insertion, the receiver first checks whether $\exists i \in [1,a]$ that the bin $h_i(x)$ is empty, where $x\in X$.
If so, it stores $x$ in one of the empty bins;
otherwise, it randomly evicts the current item in bin $h_i(x)$ and then recursively tries to insert the evicted item.
The final evicted item is placed in a stash if this process does not terminate after a certain number of trials.
If $\epsilon$ and $a$ are chosen properly, the stash can be removed~\cite{pinkas2019efficient, chandran2022circuit, rindal2021vole}.
For \emph{Simple hashing} insertion, the sender hashes $y \in Y$ to bins by \emph{all} $h_i$; each bin may have multiple items. 
Both parties pad their tables with ``fake entries'' to the maximum bin size.


There is \emph{exactly} one bin to which both parties map each common item.
PSI is then converted to private set membership testing:
Whether the item that the receiver places in a bin is among the items that are placed in this bin by the sender.
Oblivious programmable pseudorandom function is employed to obtain the final intersection result in the form of boolean shares, which involves $\mathcal{O}(1)$ secure comparisons. Therefore, there are $\mathcal{O}(n)$ secure comparisons in circuit-PSI.
The subsequent operations are applied to the output shares, the size is equal to the bin size, \ie, related to the record size of the receiver.
Thus, the record size of the sender does not affect the description of the subsequent protocols, which indicates that the assumption of $n_0=n_1=...=n$ is reasonable.
Although circuit-PSI enables secure computation of the function $\mathsf{f}$ over the intersection, it only supports symmetric functions that are oblivious to the order of the elements in the intersected set.

\subsubsection{Secure Permutation}

The secure permutation protocol 
enables two parties, $P_0$ and $P_1$, to permute the input shares, $\langle\mathbf{x}
\rangle$, based on a permutation $\pi$ defined by one of the parties. The output remains shared as $\langle\mathbf{y}
\rangle$.
Two techniques can instantiate it: Homomorphic Encryption (HE)~\cite{chase2020secret, fang2021large} and Oblivious Switching Network (OSN)~\cite{mohassel2013hide, chase2020secret} .
HE-based approaches have a $\mathcal{O}(n)$ complexity but require costly asymmetric encryption. While OSN relies on efficient symmetric key operations and OT, but has a $\mathcal{O}(n\log n)$ complexity.
Motivated by the efficiency advantages it offers over the HE-based approaches, as reported in~\cite{mohassel2013hide}, we adopt the OSN-based Oblivious Extended Permutation (OEP)~\cite{mohassel2013hide} as our baseline approach. 
The OEP is capable of handling the secure permutation problem with an extended permutation $\pi$, which allows for the replication or omission of elements.
Formally, given $N, M \in \mathbb{N}+, N\leq M$, the permutation $\pi$ is considered an extended permutation if, for every $y\in[1,N]$, there exists exactly one $x\in[1, M]$ such that $\pi(x) = y$.

\begin{figure*}[t]
    \centering
    \includegraphics[width=\linewidth]{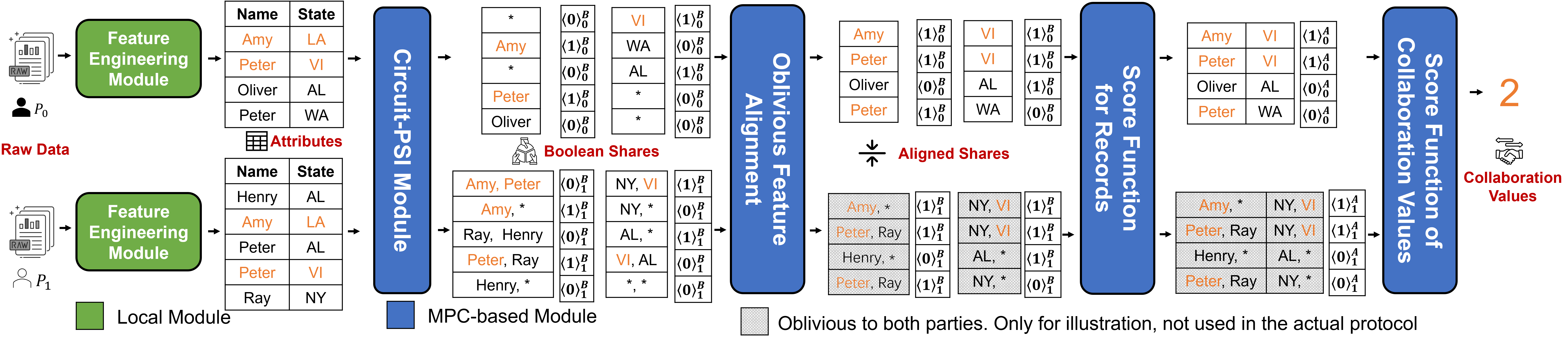}
    \vspace{-5pt}
    \caption{An example of \sys. The records marked in orange right after the feature engineering module are linked.} 
    \label{fig:overview_example}
    \vspace{-5pt}
\end{figure*}

\section{Privacy-preserving Record Screening} 
For simplicity, $P_0$ and $P_1$ are the two participants.




\subsection{Overview of \sys}

To realize the idea from Section~\ref{sec:screeing4rl} in a privacy-preserving manner, we proposed MPC-based PPRS, \sys.
Fig~\ref{fig:overview_example} gives an example to illustrate the process:

\noindent$\bullet$ \textbf{Feature Engineering Module}:
Each party $P_i$ locally preprocesses its dataset $\mathbf{V}^i$ to $\bar{\mathbf{V}}^i$.
To support schema-aware setting, derived attributes are generated empirically~(Section.~\ref{sec:s4rl_schema_aware}).
To support the approximate match, one attribute is expanded to $B$ band signatures via LSH (Section.~\ref{sec:design_non_exact_matching}).
For simplicity, we assume both $\bar{\mathbf{V}}^0$, $\bar{\mathbf{V}}^1$ contains $n$ records and $m$ derived attributes (or $Bm$ band signatures).


\noindent$\bullet$ \textbf{Circuit-PSI Module}:
It takes the derived attributes and identifies the matched ones (in the intersection) in a privacy-preserving manner. 
The result is in boolean secret shares to ensure no information leakage of $\mathbf{V}^i$.


\noindent$\bullet$ \textbf{Oblivious Feature Alignment (OFA)}:
Similar to the plaintext idea in~Section~\ref{sec:screeing4rl}, 
one instance of circuit-PSI will be invoked for each derived attribute.
Due to the hashing operation in circuit-PSI, the attribute of the same records will be out-of-order.
The OFA protocol addresses this issue by employing the secure permutation technique with several optimizations. 
It ensures the elements corresponding to the same record are properly aligned for subsequent computations.


\noindent$\bullet$ \textbf{Score Function Module}:
It determines whether the record is linked or not obliviously, ultimately deriving the final collaboration value between two parties' datasets.

\subsection{Circuit-PSI Module}

In a na\"ive exact matching and schema-agnostic setting, PSI could easily handle. However, PSI can support neither approximate matching nor schema-aware setting without information leakage.
\sys\ runs multiple circuit-PSI instances, each for one attribute (to support schema-aware) or one band (to support LSH bin for the approximate matching), and the privacy of the intermediate result can be preserved.
Specifically, 
two parties proceed to carry out $m$ ($Bm$ for approximate matching) separate instances of the circuit-PSI protocol, with each instance aligning with a specific derived attribute (\eg, $j$-th attribute) from $\bar{\mathbf{V}}^0$ and $\bar{\mathbf{V}}^1$.
In an individual circuit-PSI instance, $P_0$ (resp. $P_1$) acts as the sender 
 (resp. receiver), and inputs $j$-th attribute value $\bar{\mathbf{V}}^0$  (resp. $\bar{\mathbf{V}}^1$) of all records. 
It generates boolean shares $\langle \mathbf{q}_j \rangle_l^B$ for $P_l$, which indicates the presence of $j$-th derived attribute within the intersection set.



Both the Cuckoo hash and simple hash 
can be conceptualized as permutation functions, denoted as $\pi^C$ and $\pi^S$, respectively. 
It causes two problems:
(i) misalignment of intersection results across attributes: the order of the $\langle \mathbf{q}_j \rangle^B$ is determined by the Cuckoo hash, which is different for different attributes;
(ii) the result of the duplicate item is missing: the hash operation maps duplicate attribute values to the same hash bin.
Therefore, to align $\langle \mathbf{q}_j \rangle^B$ of all the attributes, we need to apply the extended permutation, ${\pi^C_j}^{-1}:[0,(1+\epsilon)n)\rightarrow [0, n)$, on $\langle \mathbf{q}_j \rangle^B$ obliviously.
The problem is not trivial due to the simple hash maps multiple records residing within the same bin, \ie, the order of records after circuit-PSI is not disclosed to $P_1$. 
A straightforward solution is to let $P_0$ send the permutation $\pi^C$ to $P_1$. Nevertheless, this raises concerns regarding the potential information leakage related to $\bar{\mathbf{V}}^0$, as in~\Cref{fig:problem_send_permutation}.
By analyzing $\pi^C$, $P_1$ can infer the likelihood of a record matching. For example, if the attribute values of a record aligned together post-alignment, the probability of a match, $\Pr(\mbox{Match})$, could be substantially elevated.

\begin{figure}[t]
\centering
\includegraphics[width=\linewidth]{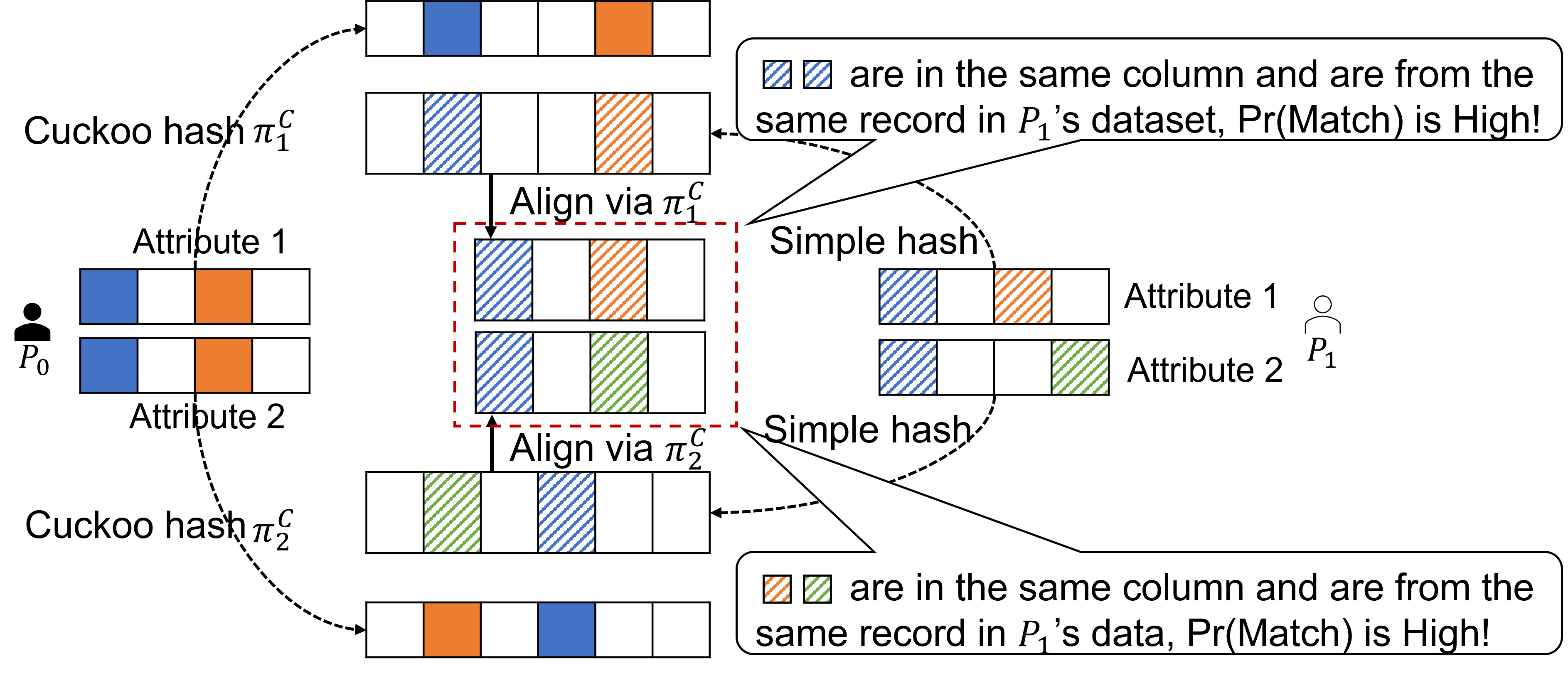}
\caption{An example of the information leakage if $P_0$ transmits the permutation $\pi^c$ to $P_1$. The bin value with the same color is from the same record.} 
\label{fig:problem_send_permutation}
\end{figure}

\begin{figure}[t]
    \centering
    \includegraphics[width=0.68\linewidth]{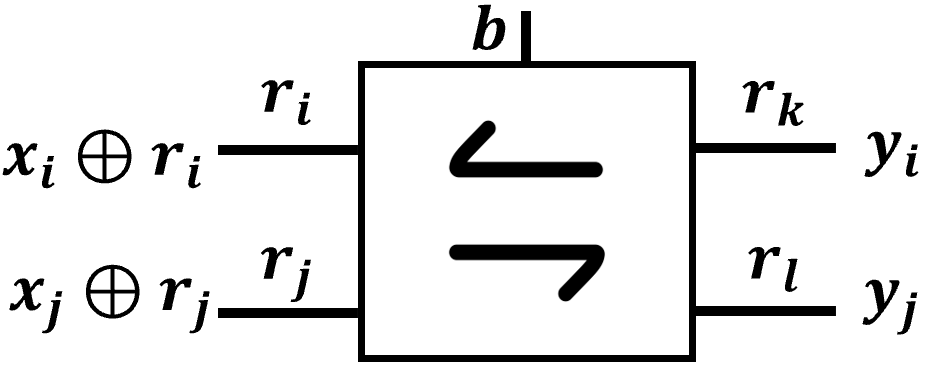}
    \caption{The structure of 1-switch. $r_i, r_j, r_k, r_l$ are the random wire labels generated by the receiver, $b$ is the selection bit from the sender, $x_i$ and $x_j$ are receiver's input, $y_i$ and $y_j$ are the blinded output to the sender.}
    \label{fig:1_switch}
\end{figure}

\subsection{Oblivious Feature Alignment (OFA) Protocol}

We abstract the above issue as a~\emph{secure permutation} problem: $P_0$ inputs the extended permutation ${{\pi^C_j}^{-1}}$ and the  $\langle \mathbf{q}_j \rangle_0^B$, while $P_1$ inputs $\langle \mathbf{q}_j \rangle^B_1$. The protocol provides each party the boolean shares post-permutation, denoted as $\langle \bar{\mathbf{q}}_j \rangle^B={\pi^c_j}^{-1}(\langle \mathbf{q}_j \rangle^B)$.
To solve this, we propose the Oblivious Feature Alignment (OFA) protocol, which tailors the OEP protocol~\cite{mohassel2013hide} to meet our requirements.


\newcolumntype{M}[1]{>{\centering\arraybackslash}m{#1}}
\renewcommand{\arraystretch}{1.2}
\begin{table}[t]
    \centering
    \caption{Output $y_i, y_j$ with different selection bit $b$ in replication 1-switch and permutation 1-swtich. The value marked in red can be removed in OFA for optimization.}
    \begin{tabular}{|c|c|c|c|c|c|c|}
        \hline 
        Switch Type & $b$ & $y_i$ &$y_j$ & $\mathbf{T}[b,i]$ &  $\mathbf{T}[b,j]$ \\ 
        \hline \hline
        \multirow{2}{*}{\begin{tabular}{@{}c@{}}Permutation \\ 1-switch\end{tabular}} & $0$ &  $x_i\oplus r_k$ & $x_j\oplus r_l$ & $r_i \oplus r_k$ & $r_j\oplus r_l$\\  
        \cline{2-6}
        & $1$ & $x_j\oplus r_k$ & $x_i\oplus r_l$ & $r_j\oplus r_k$&  $r_i\oplus r_l $ \\
        \hline \hline
        \multirow{2}{*}{\begin{tabular}{@{}c@{}}Replication \\ 1-switch\end{tabular}} & $0$ & $x_i\oplus r_k$  & $x_j\oplus r_l$ & $\color{red}{r_i \oplus r_k}$ & $r_j \oplus r_l$  \\  
        \cline{2-6}
        & $1$ & $x_i\oplus r_k$ &  $x_i\oplus r_l$  & $\color{red}{r_i\oplus r_k}$ & $r_i\oplus r_l$ \\
        \hline 
    \end{tabular}
    \label{tab:output_table_1_switch}
\end{table}

\subsubsection{OEP~\cite{mohassel2013hide}}
OEP is constructed via 2 basic protocols - Oblivious Replication Network (ORN) and Oblivious Permutation Network (OPN). 
The former can obliviously replicate the adjacent elements within a vector, while the latter can obliviously reorder the elements of a vector by giving a bijection mapping (\ie, one-one and onto).
Both OPN and ORN are built on the 1-switch, which replicates/permutes the adjacent input value for both parties based on a selection bit, as Fig.~\ref{fig:1_switch} shows.
The difference is the input-output table (shown in~\ref{tab:output_table_1_switch}) and network topology of these 1-switches.
OPN is implemented via a topology of Bene\v{s} Network containing $n\log n -n +1$ permutation switch, while ORN is implemented by sequentially connecting $n-1$ replication switch. 

For a given 1-switch, $P_1$ assigns random blinded values $r_i$, $r_j$ (resp. $r_k$, $r_l$) to the input (resp. output) wires, and inputs $x_i$ and $x_j$.
They get the output shares of $y_i$ and $y_j$,  determined by the $P_0$'s selection bit $b$.
All the selection bits are derived from the extended permutation $\pi: [1, N]\rightarrow [1, M]$ via depth-first search (DFS).
A 1-switch is realized by one call of $\mathcal{F}_{\mbox{OT}}$. 
Let's take replication 1-switch with $b=1$, meaning to copy the value of $y_j$ from $y_i$, as an example:
First,  $P_1$ blinds its input $\langle x_i \rangle^A_0 = x_i \oplus r_i$ and $\langle x_j \rangle^A_0 = x_j\oplus r_j$, and sends them to $P_0$.
Next, $P_1$ creates table $\mathbf{T}$ as Table~\ref{tab:output_table_1_switch} shows.
Then both parties invoke $\mathcal{F}_{\mbox{OT}}$, where $P_0$ inputs $b$ 
and $P_1$ inputs $\mathbf{T}$.
Finally, $P_0$ retrieves $\mathbf{T}[b]$ from $\mathcal{F}_{\mbox{OT}}$, and outputs $y_i = \langle x_i \rangle^A_0 \oplus \mathbf{T}[b,i]=x_i \oplus r_k$ and  $y_j = \langle x_i \rangle^A_0 \oplus  \mathbf{T}[b,j] = x_i \oplus r_l$.
Both are the shares of replication value $x_i$ because $P_1$ holds the other share $r_l$ and $r_k$.
This can be easily extended to the case where both parties input the shares:
$P_1$ replaces $x_i$,$x_j$ with its input shares, $P_0$ XOR the result $y_l$, $y_k$ with 
its reordered (permuted/replicated) input shares according to its $b$.

Built on OPN and ORN, OEP (and the improved OFA) has three following phases, which take $N$ \emph{real inputs} with $M-N$ \emph{dummy inputs} as Fig.~\ref{fig:OFA_workflow} shows:

\noindent $\bullet$ \emph{Dummy-value Placement Phase} (OPN):
For each real input that is mapped to $e$ different output positions according to $\pi$, it outputs real value followed by $e-1$ dummy values, which serve as a placeholder for replication in the subsequent phase.


\noindent $\bullet$ \emph{Replication Phase} (ORN): 
Taking the output from the previous phase as input, it directly outputs the real input and replaces the dummy input with the real input that precedes it. 


\noindent $\bullet$ \emph{Permutation Phase} (OPN). 
The output from the previous phase is taken and permuted to the final location under the guidance of $\pi$, 
which is in the same order as the $P_0$'s records.

\subsubsection{Optimizations in OFA}
OEP supports the general extended permutation, we take a further step and tailor it for the permutation of Cuckoo hash (${{\pi^C_j}^{-1}}$) with three optimizations (shown in Fig.~\ref{fig:OFA_workflow}), which 
significantly improve the efficiency.


\begin{figure*}
    \centering
    \includegraphics[width=\linewidth]{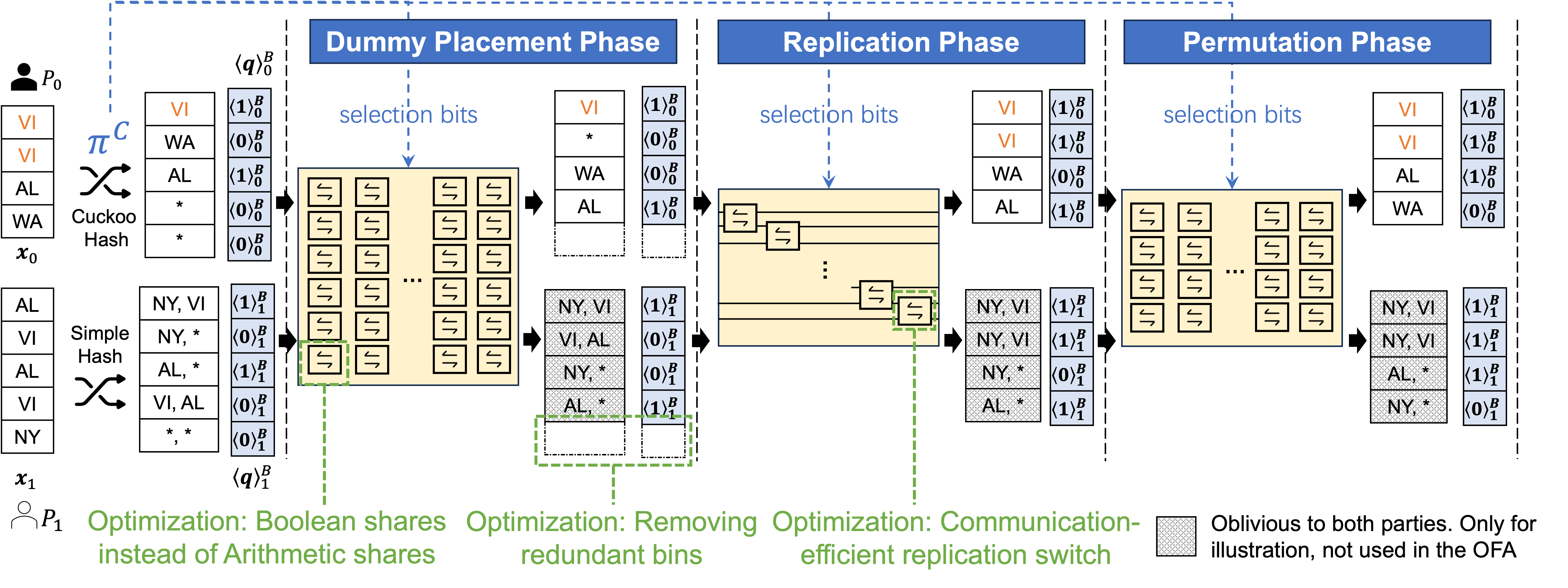}
    \caption{The workflow of Oblivious Feature Alignment (OFA) protocol. ``*'' represents a random fake entry, the attribute value marked in orange is the repeated input. $P_0$ (resp. $P_1$) inputs the value $\mathbf{x}_0$ (resp. $\mathbf{x}_1$) to circuit-PSI and gets the boolean shares result $\langle \mathbf{q}\rangle^B$ indicates the intersection result, \ie, $\mathbf{x}_0 \cap \mathbf{x}_1$. The dummy placement phase reorders the dummy values ``*'' (resp. ``VI, AL'') to the second bin right after the bin with ``VI'', which will be used in replication later. It also permutes all other redundant bins to the bottom, which will be then dropped by both parties. The replication phase copies the boolean shares of dummy inputs from the boolean share represented ``VI''. Finally, the permutation phase permutes all the bins to the order of $P_0$'s records before input to circuit-PSI. The bins (\ie, boolean shares) with blue color are the input and output of each phase. The bins of the attribute values, State, are only for illustration and are not used in the actual protocol.}
    \label{fig:OFA_workflow}
    \vspace{-5pt}
\end{figure*}

\noindent $\bullet$ \emph{Optimization 1 - Dummy-value Placement Phase}:
The input of secure permutation containing $\epsilon n$ redundant bins for privacy (see Section~\ref{sec:preliminaries_circuit_psi}), $e'$ dummy inputs, and $n-e'$ real inputs. 
Thus, during deriving the selection bits from the permutation ${\pi^C_j}^{-1}$ in DFS, we add a constraint that orders the $\epsilon n$ redundant bins at the end of the output values.
Both parties can now easily drop the last $\epsilon n$ output values, 
reducing the 1-switch used for the subsequent two phases, ultimately decreasing the number of $\mathcal{F}_{\mbox{OT}}$
from $2(1 + \epsilon)n \log((1+ \epsilon)n) - (1 + \epsilon) n + 1$ to $(1 + \epsilon)n \log((1 + \epsilon) n) + n \log n - (1 + \epsilon) n + 1$.

\noindent $\bullet$ \emph{Optimization 2 - Replication Phase}: 
We found that in replication 1-switch, the top output ($y_i$) is always equal to the top input ($x_i$) in Fig.~\ref{fig:1_switch}, regardless of the selection bit. 
Therefore, $P_1$ only needs to generate and send a partial table $\mathbf{T}[b,j]$ to $\mathcal{F}_{\mbox{OT}}$, omitting $\mathbf{T}[b,i]$ (marked in red in Table~\ref{tab:output_table_1_switch}) related to the top output wire ($r_k$). 
$P_1$ can directly blind input $x_i$ with top output label $r_k$ instead of $r_i$, \ie, $y'_i=x_i\oplus r_k$, and $P_0$ can output $y_i = y'_i$.
Therefore, we halve the communication and computation cost in replication 1-switch.

\noindent $\bullet$ \emph{Optimization 3 - Shares Type}:
OEP is designed for arithmetic shares, which typically requires a larger bits (\eg, 128 bits) in the circuit. By observing  $\langle \mathbf{q}_j \rangle^B$ are 
boolean shares, we replace all the random arithmetic numbers $r$ (16 bytes) used in the 1-switch with random boolean values (one byte).
Since these random values is used in $\mathcal{F}_{\mbox{OT}}$ and $\mathbf{T}$, 
we achieve a substantial $16\times$ reduction in the communication cost. 


\subsubsection{Communication-efficient OFA}
With the integration of the three aforementioned optimizations, 
we present our highly efficient OFA protocol.
The workflow is shown in Fig.~\ref{fig:OFA_workflow}, and the functionality and detailed protocol are shown in Fig.~\ref{fig:OFA}

\algnewcommand\algorithmicsenderinput{\textbf{$P_0$'s Input:}}
\algnewcommand\SenderInput{\item[\algorithmicsenderinput]}
\algnewcommand\algorithmicreceiverinput{\textbf{$P_1$'s Input:}}
\algnewcommand\ReceiverInput{\item[\algorithmicreceiverinput]}
\algnewcommand\algorithmicpublicinput{\textbf{Public Input:}}
\algnewcommand\PublicInput{\item[\algorithmicpublicinput]}

\begin{figure}[t!]
\begin{tcolorbox}[colback = white!5!white, left = 0mm, right = 0mm, top = 0mm, bottom = 0mm, title = Functionality $\mathcal{F}_{\mbox{OFA}}$]
\begin{algorithmic}[1]
	\SenderInput Boolean shares $\langle \mathbf{q} \rangle^B_0$ of size $(1+\epsilon)n$, the extended permutation mapping $\pi:[0, (1+\epsilon)n)\rightarrow[0,n)$
        \ReceiverInput Boolean shares $\langle \mathbf{q} \rangle^B_1$ of size $(1+\epsilon)n$
        \Ensure $P_l$ ($l=0/1$) learns $\langle \bar{\mathbf{q}} \rangle^B_l$ of size $n$, s.t. $\bar{\mathbf{q}}=\pi(\mathbf{q})$.
\end{algorithmic}
\end{tcolorbox}

\begin{tcolorbox}[colback = white!5!white, left = -1mm, right = 0mm, top = 0mm, bottom = 0mm, title = Protocol $\varPi_{\mbox{OFA}}$]
    \begin{algorithmic}[1]
        \Statex \Comment{$P_0$ derives the selection bits $\mathbf{b}$ of $\pi$ and learns a random encoding of the whole network with $\mathbf{b}$.}
        \State $P_l$: constructs the network of three phases via 1-switch.
        \State $P_0$: derives $\mathbf{b}$ of these switches via DFS.
        \For {network of dummy-value placement phase (OPN), replication phase (ORN), permutation phase (OPN)}
        \For {each wire $i$ of the network} \label{fig:OFA_L4}
        \State $P_1$: generates a random bit $r_i \sample \{0,1\}$.
        \EndFor 
        \For{each switch $i$ of the network}
            \State $P_1$: creates and sends table $\mathbf{T}$ as Table~\ref{tab:output_table_1_switch}  to $\mathcal{F_{\mbox{OT}}}$.
            \State $P_0$: sends $\mathbf{b}[i]$ to $\mathcal{\mathcal{F}_{\mbox{OT}}}$ and gets $\mathbf{T}[\mathbf{b}[i]]$. \label{fig:OFA_L8}
        \EndFor
        \EndFor
        \Statex \Comment{$P_1$ blinds its input shares}
        \State $P_1$: sends $y'_i = \langle \mathbf{q} \rangle^B_1[i] \oplus r_i$ to $P_0$ for all input wire $w_i$. \label{fig:OFA_L9}
        \Statex \Comment{$P_0$ evaluates the random encoding of the whole network on $P_1$’s blinded inputs.}
        \State $P_0$:  $y_i = y'_i  \oplus \langle \mathbf{q}_i \rangle^B_0=\mathbf{q}[i] \oplus r_i$ for all input wire $w_i$.
        \State In topological order, for each switch with selection bit $b$, input wires $w_i, w_j$ and output wires $w_k, w_l$, $P_0$ does:
        \If {switch is a replication switch}
        \If{$b=0$}  $y_k=y_i, y_l = y_j \oplus \mathbf{T}[b,j]$ 
        \Else~$y_k=y_i, y_l=y_i\oplus \mathbf{T}[b,j]$
        \EndIf
        \EndIf
        \If {switch is a permutation switch}
        \If{$b=0$} $y_k= y_i\oplus \mathbf{T}[b,i], y_l = y_j \oplus \mathbf{T}[b,j]$ 
        \Else~$y_k= y_j \oplus \mathbf{T}[b,j], y_l = y_i \oplus \mathbf{T}[b,i]$ \EndIf
        \EndIf
        \Statex \Comment{$P_0$ unblinds and retrieves the output}
        \State $P_0$: blinds $y_k$ with its output shares for each output switch $h (0\leq h < n/2)$ with output wires $w_k, w_l$:  $\langle\bar{\mathbf{q}}\rangle_0^B[2h]=y_k \oplus  \pi(\langle\mathbf{q}\rangle^B_0)[2h], \langle\bar{\mathbf{q}}\rangle_0^B[2h+1]=y_l\oplus\pi(\langle\mathbf{q}\rangle^B_0)[2h+1]$.
        \State $P_1$: outputs the randoms $r_k,r_l$ of output wires $w_k, w_l$ in the network of permutation phase, $\langle  \bar{\mathbf{q}} \rangle^B_1[2h]=r_k, \langle \bar{\mathbf{q}} \rangle^B_1[2h+1]=r_l$ for the $h$-th output switch. \label{fig:OFA_L19}
    \end{algorithmic}
\end{tcolorbox}
\caption{Ideal functionality and protocol of OFA}
\label{fig:OFA}
\end{figure}

\begin{figure}[t!]
\begin{tcolorbox}[colback = white!5!white, left = 0mm, right = 0mm, top = 0mm, bottom = 0mm, title = Functionality $\mathcal{F}_{\mbox{\sys}}$]
\begin{algorithmic}[1]
	\SenderInput Database $\mathbf{V}^0$ with $n$ records.
        \ReceiverInput Database $\mathbf{V}^1$  with $n$ records.
        \PublicInput $B, R$ in LSH, threshold $t$ in score function for records.
        \Ensure $P_l$ ($l=0/1$) gets the collaboration value $\langle c \rangle^B_l$.
\end{algorithmic}
\end{tcolorbox}

\begin{tcolorbox}[colback = white!5!white, left = -1mm, right = 0mm, top = 0mm, bottom = 0mm, title = Protocol $\varPi_{\mbox{\sys}}$]
    \begin{algorithmic}[1]
        \Statex \Comment{Feature Engineering Module}
        \State $P_l$ locally preprocess $\mathbf{V}^l$ to generate $\mathbf{\bar{V}}^l$, which contains $n$ records and $Bm$ derived attributes. Specifically, it can derive new attributes via concatenate attributes to support schema-aware (Section.~\ref{sec:s4rl_schema_aware}), or can apply LSH to expand an attribute to multiple band signatures for approximate matching (Section.~\ref{sec:design_non_exact_matching}).
        \Statex \Comment{Circuit-PSI}
        \For{each derived attribute $j$ in $\mathbf{\bar{V}}^l$}
        \State $P_0$ takes role of receiver, $P_1$ takes role of sender, invokes $\mathcal{F}_{\mbox{CPSI}}$, $P_l$ gets the Boolean shares indicates the presence of the attribute value in  the intersection $\langle\mathbf{q}_j \rangle^B_l=\mathcal{F}_{\mbox{CPSI}}(\{\bar{\mathbf{V}}^b[:,j]\},\{\bar{\mathbf{V}}^1[:,j]\})$, $P_0$ also gets the mapping of Cuckoo hash $\pi^C_j$. \label{alg:functionality_sys:L4}
        \EndFor
        \Statex \Comment{Oblivious Feature Alignment}
        \For{each derived attribute $j$ in $\mathbf{\bar{V}}^l$}
        \State Both parties jointly invoke $\mathcal{F}_{\mbox{OFA}}$ and $P_l$ gets the aligned attribute intersection result $\langle\bar{\mathbf{q}}_j \rangle^B_l=\mathcal{F}_{\mbox{OFA}}(\{\langle \mathbf{q} \rangle^B_0, {\pi^C_j}^{-1} \},\{\langle \mathbf{q} \rangle^B_1\})$
        \EndFor
        \Statex \Comment{Score Function Module}
        \For {each record $i$}
        \State $P_l$ calculates the decision of linkage $\langle \mathbf{d} \rangle^B_l[i]$ according to Section~\ref{sec:score_records} via $\mathcal{F}_{\mbox{MUX}}$, $\mathcal{F}_{\mbox{MSB}}$, $\mathcal{F}_{\mbox{AND}}$
        \EndFor
        \State $P_l$ gets collaboration value $\langle c \rangle^A_l=\sum_i \mathcal{F}_{\mbox{B2A}}(\langle \mathbf{d} \rangle^B_l)[i]$
    \end{algorithmic}
\end{tcolorbox}
\caption{Ideal functionality and protocol of \sys}
\label{fig:functionality_sys}
\end{figure}

\subsection{Score Function Module}

After executing OFA, all intersection results, $\langle \bar{\mathbf{q}} \rangle^B$, across various attributes are aligned to match the original order of records in $\mathbf{V}^0$.
Then, \sys\ can determine whether each record is linked via a score function, and calculate the value of the collaboration between the two parties. 


\subsubsection{Scoring Function for Records}
\label{sec:score_records}
\sys\ accommodates various polynomial score functions, providing flexibility in determining how records are linked. Below are two exemplary scoring models supported by \sys:

\noindent $\bullet$ \textbf{Simple linear model}. In this model, a record is linked iff all its attributes are matched. 
The matching result for the $i$-th record can be determined using  $\mathcal{F}_{\mbox{AND}}$ across all attribute $\langle \mathbf{d} \rangle^B[i] = \bigwedge_j{\langle \bar{\mathbf{q}}_j \rangle^B}[i]$, where $j$ iterates over all attributes.

\noindent $\bullet$ \textbf{General linear/logistic regression model}. 
A linear/logistic regression is applied to score the records. 
It assigns different attributes different weights, \ie, the more selective an attribute is, the more it contributes to the matching decision.
A record is determined to be linked by comparing the score with a pre-determined threshold $t$. 
The threshold and the weights can be decided by an offline machine learning algorithm trained on public/private data. 
The specific workflow shows as follows. (i) Weight Sharing: $P_0$ generates the arithmetic shares of the weights, and sends one share to $P_1$;
(ii) Score Calculation: Both parties compute the score for the $i$-th record as $\langle \mathbf{s} \rangle^A[i] = \sum_j (\mathcal{F}_{\mbox{MUX}}(\langle w^e_j \rangle^A, \langle \bar{\mathbf{q}}_j \rangle^B[i]) + \mathcal{F}_{\mbox{MUX}}(\langle w^n_j \rangle^A, \neg \langle \bar{\mathbf{q}}_j \rangle^B[i]))$, where $w^e_j$ denotes the weight of the $j$-th attribute when it is matched, and $w^n_j$ for unmatched;
(iii) Decision Making: The linkage decision for the $i$-th record is determined by comparing the score with the threshold, $\langle \mathbf{d} \rangle^B[i] = \mathcal{F}_{\mbox{MSB}}(\langle \mathbf{s} \rangle^A[i] - \langle t \rangle^A)$.

The former model offers a faster method, while the latter is more fine-grained to consider varying importances of different attributes. 
Both models are executed securely through MPC. 

\subsubsection{Score Function for Collaboration Value}
The collaboration value $c$ is a statistical aggregated value over the results of linked records. 
Finding an appropriate scoring function is an open question and depends on the specific applications.
This is orthogonal to \sys, which focuses on providing a flexible framework to support different scoring functions. 
For simplicity, we adopt cardinality counting for scoring collaboration value, \ie, $\langle c \rangle^A=\sum_i{\mathcal{F}_{\mbox{B2A}}(\langle \mathbf{d} \rangle^B[i])}$.

\subsection{Practical Consideration}

\subsubsection{Handling Missing Value}
\label{sec:design_practical_missing_value}
In many applications, data is incomplete,~\ie, some attribute values are missing.
However, the output of circuit-PSI is restricted to only two cases:
match or unmatch. 
This causes the attribute with the missing value to eventually be deemed a match, resulting in an elevated false positive rate.
Note that the receiver ($P_0$) in circuit-PSI is aware of which values are missing, as it inserts the data using Cuckoo hashing and has access to $\pi^C$. It can adjust the score of the records having the missing attribute value.
Taking the linear/logistic model as an example, we denote the weight for missing values as $w^m_j$ for attribute $j$. After both parties have computed $\langle \mathbf{s}\rangle^A$, if attribute value $\bar{\mathbf{V}}[i,j]$ of record $i$ is absent, $P_0$ can adjust the score via $\langle \mathbf{s}\rangle^A_0[i] = \langle \mathbf{s}\rangle^A_0[i] + w^m_j$.

\subsubsection{Supporting Approximate Matching}
\label{sec:design_non_exact_matching}
Section~\ref{sec:screeing4rl} introduces to support the approximate matching via LSH.
To incorporate it in \sys, we modify the following: 
(i) $P_0$ and $P_1$ synchronize parameters $B$, $R$;
(ii) They extract $q$-grams for each attribute values;
(iii) Each party creates a new $\bar{\mathbf{V}}$ by expanding one attribute to $B$ band signatures (detailed in Section~\ref{sec:screeing4rl}), increasing the number of derived attributes from $m$ to $Bm$;
(iv) Both parties invoke the circuit-PSI per attribute, \ie, $Bm$ instances of circuit-PSI are executed in total;
(v) OFA is used to align the intersection result of all derived attributes.
(vi) In the score function module, the original $j$-th attribute (in $\mathbf{V}$) of record $i$ is linked iff 
one of its $B$ derived attributes (band signatures) is matched
$\langle \mathbf{\bar{q}'}_j \rangle^B[i]= \bigvee_{k=j\cdot B}^{(j+1)B-1}\langle \bar{\mathbf{q}}_{k}\rangle[i]$, which can be securely realized via $\mathcal{F}_{\mbox{OR}}$. And the decision $\mathbf{d}$ of the record is linked is calculated by $\langle \mathbf{\bar{q}'}_j \rangle^B$.

\subsection{Putting All Things Together}
By assembling all the modules, we get \sys, the ideal functionality and the protocol is shown in Fig.~\ref{fig:functionality_sys}. 
\section{Analysis}

\subsection{Security Analysis}
In this section, we provide sketches of proof that OFA and \sys\ are secure according to~\Cref{def:2pc}.



\begin{theorem}
\label{theorem:OFA}
$\varPi_{\mbox{OFA}}$ is a secure protocol under $\mathcal{F}_{\mbox{OT}}$-hybrid.
\end{theorem}
\begin{proof}

\textbf{Case 1:  $P_0$ is corrupted.}

The simulator $\mathcal{S}$ constructs the network of three phases via 1-switch and interacts with $P_0$ as follows. 

\noindent \textbf{Step 1}. $\mathcal{S}$ assigns random bits for each wire in the network of the dummy-value placement phase and constructs the table $\mathbf{T}$ of each permutation switch. Then, $\mathcal{S}$ works as sender in $\mathcal{F}_{\mbox{OT}}$. When $P_0$ sends selection bit $b$ to $\mathcal{F}_{\mbox{OT}}$, $\mathcal{S}$ sends $\mathbf{T}$ to $\mathcal{F}_{\mbox{OT}}$. 

\noindent \textbf{Step 2}. $\mathcal{S}$ assigns random bits for each wire in the network of the replication phase and constructs the table $\mathbf{T}$ of each replication switch. Then, $\mathcal{S}$ works as sender in $\mathcal{F}_{\mbox{OT}}$. When $P_0$ sends selection bit $b$ to $\mathcal{F}_{\mbox{OT}}$, $\mathcal{S}$ sends $\mathbf{T}$ to $\mathcal{F}_{\mbox{OT}}$. 

\noindent \textbf{Step 3}. $\mathcal{S}$ assigns random bits of each wire in the network of the permutation phase and constructs the table $\mathbf{T}$ of each permutation switch. Then, $\mathcal{S}$ works as sender in $\mathcal{F}_{\mbox{OT}}$. When $P_0$ sends selection bit $b$ to $\mathcal{F}_{\mbox{OT}}$, $\mathcal{S}$ sends $\mathbf{T}$ to $\mathcal{F}_{\mbox{OT}}$. 

\noindent \textbf{Step 4}. $\mathcal{S}$ XOR the input vector $\langle \mathbf{q}\rangle^B_1$ with random bits in Step 1 and sends to $P_0$, and outputs the random bits of output wires in the permutation phase (OPN). 

To prove that this simulation is indistinguishable from the real protocol, we consider the following hybrid worlds. 

\noindent $\mathbf{Hybrid_0}$ The same as the real protocol. 

\noindent $\mathbf{Hybrid_1}$ Each permutation switch in dummy-value placement phase (Line~\ref{fig:OFA_L4}--~\ref{fig:OFA_L8}, $\varPi_{\mbox{OFA}}$) is replaced by $\mathcal{F}_{\mbox{OT}}$

\noindent $\mathbf{Hybrid_2}$ 
Each replication switch in replication phase (Line~\ref{fig:OFA_L4}--~\ref{fig:OFA_L8}, $\varPi_{\mbox{OFA}}$) is replaced by $\mathcal{F}_{\mbox{OT}}$ 

\noindent $\mathbf{Hybrid_3}$ Each permutation switch in permutation phase (Line~\ref{fig:OFA_L4}--~\ref{fig:OFA_L8}, $\varPi_{\mbox{OFA}}$) is replaced by $\mathcal{F}_{\mbox{OT}}$.

\noindent $\mathbf{Hybrid_4}$ 
The execution of $\mathcal{S}$ in step 4 is identical to Line~\ref{fig:OFA_L9}--\ref{fig:OFA_L19} in $\varPi_{\mbox{OFA}}$. 

Note that $\mathbf{Hybrid_1}$, $\mathbf{Hybrid_2}$ and $\mathbf{Hybrid_3}$ is computationally indistinguishable from real protocol as $\mathcal{F}_{\mbox{OT}}$ is secure. Therefore, the simulation ($\mathbf{Hybrid_4}$) is computationally indistinguishable from the view of real protocol ($\mathbf{Hybrid_0}$).


\noindent \textbf{Case 2:  $P_1$ is corrupted.}

The simulator $\mathcal{S}$ constructs the network of three phases via 1-switch and derives $\mathbf{b}$ of these switches via DFS. Then, the interaction between $\mathcal{S}$ and $P_1$ is similar to the situation when $P_0$ is corrupted, except $\mathcal{S}$ works as receiver in $\mathcal{F}_{\mbox{OT}}$. 

To sum up, in both cases, The $\mathcal{F}_{\mbox{OT}}$-hybrid OFA protocol is secure in the semi-honest model.
\end{proof}

\begin{theorem}
\label{thm:appraisal}
The $\mathcal{F}_{\mbox{OT}}$-hybrid $\varPi_{\mbox{\sys}}$ is secure. 
\end{theorem}
\begin{proof}[Proof sketch]
$\varPi_{\mbox{\sys}}$ consists of $mB$ $\mathcal{F}_{\mbox{CPSI}}$, $mB$ $\mathcal{F}_{\mbox{OFA}}$, 
\emph{Score Function} including $\mathcal{F}_{\mbox{MUX}}$, $\mathcal{F}_{\mbox{MSB}}$, $\mathcal{F}_{\mbox{AND}}$, and $\mathcal{F}_{\mbox{B2A}}$. 
All of these functionalities are secure 2PC $\mathcal{F}_{\mbox{OT}}$-hybrid protocols in the semi-honest model. 
Thus, we can construct a simulator for \sys~using the simulators of $\mathcal{F}_{\mbox{CPSI}}$, $\mathcal{F}_{\mbox{OFA}}$, $\mathcal{F}_{\mbox{MUX}}$, $\mathcal{F}_{\mbox{MSB}}$, $\mathcal{F}_{\mbox{AND}}$ and $\mathcal{F}_{\mbox{B2A}}$, which are similar to the proof of~\Cref{theorem:OFA} and conclude that \sys~is also secure given Theorem~\ref{theorem:composition}. 
\end{proof}

\subsection{Performance Analysis}
\subsubsection{\sys} 
The analysis of time and communication complexity is as follows:

\noindent $\bullet$ \emph{Feature Engineering}. It runs locally without any communication. The time complexity is $\mathcal{O}(nm)$ for exact matching and $\mathcal{O}(Bnm)$ for approximate matching.

\noindent $\bullet$ \emph{Circuit-PSI}. The time and communication complexity of the protocol is $\mathcal{O}(n)$. For the schema-aware and approximate matching setting, it is $\mathcal{O}(nmB)$.

\noindent $\bullet$ \emph{OFA}. 
The time and communication complexity per derived attribute (and/or band signature) is $\mathcal{O}(n \log n)$ .

\noindent $\bullet$ \emph{Score Function}. The time and communication complexity per derived attribute (and/or band signature) is $\mathcal{O}(n)$.

In total, the time and communication complexity of \sys\ is $\mathcal{O}(mn \log n)$, or it is $\mathcal{O}(Bmn \log n)$ for approximate matching. Moreover, only circuit-PSI involves secure comparison, which is $\mathcal{O}(Bmn)$.

\subsubsection{Screening-then-Linkage}
The time for the pure PPRL scheme is $N\cdot t_{\mbox{PPRL}}$ where $t_{\mbox{PPRL}}$ is the time for executing the PPRL. 
The time for \emph{Screening-then-Linkage} is $N\cdot t_{\mbox{PPRS}}+ \alpha N \cdot t_{\mbox{PPRL}}$ where  $t_{\mbox{PPRS}}$ is the time for executing the PPRS and $\alpha$ is the ratio of the parties with high collaboration value. 

\section{Evaluation}

In this section, we first outline the implementation. Then, we present the evaluation setup, and show the effectiveness and efficiency of \sys\ and \emph{Screening-then-Linkage} framework.



\subsection{Implementation}
\label{sec:eva_impl}
We conducted data extraction, cleaning, and ground truth determination using Python 3 and implemented \sys\ system in C++ using the following libraries: 
the \emph{Feature Engineering} and \emph{Score Function} modules were based on Cryptflow2~\cite{rathee2020cryptflow2}, while the $\mathcal{F}_{\mbox{OT}}$ was sourced from Cheetah~\cite{huang2022cheetah}.
The \emph{Circuit-PSI} module was adapted from~\cite{chandran2022circuit}.
For the \emph{OFA} protocol, we modified the OSN code from~\cite{jia2022shuffle} and~\cite{garimella2021private}. 
We set the computational security parameter to $\lambda = 128$ and the statistical security parameter to $\sigma = 40$ 
For the stashless circuit-PSI, 
the smaller $\epsilon$, the possibility of collisions will increase;
the larger $\epsilon$, the cost of both Circuit-PSI and OFA will increase.
We use $\epsilon = 0.27$ and $a = 3$~\cite{rindal2021vole} to achieve the balance between security and performance.

\renewcommand{\arraystretch}{0.95}
\begin{table}[t!]
    \centering
    \caption{The description of the evaluation datasets.}
    \begin{tabular}{|c|c|c|c|c|c|c|}
        \hline 
        Dataset & \# of records & \begin{tabular}{@{}c@{}}Duplication (\%) \\ for both data\end{tabular} & \begin{tabular}{@{}c@{}}Matching \\ Mode\end{tabular}\\ 
        \hline
        iDash500K & 500K/500K & 10\%/10\% & Exact \\
        \hline
        iDash1M & 1,000K/1,000K & 10\%/10\% & Exact \\
        \hline
        DBLP/ACM & 2,616/2,294 & 85\%/95\% & Approximate \\
        \hline
        BNB/TPL & 3,565K/1,443K  & 17.5\%/39.5\% & Mixed \\
        \hline
    \end{tabular}
    \label{tab:eva_datasets}
\end{table}

\subsection{Evaluation Setup}
\label{sec:eva_setup}
We conducted 
our evaluations on two instances, which play the roles of $P_0$ and $P_1$ respectively. Both instances are equipped with 
10 cores of 2.5GHz Intel Xeon Platinum 8255C CPU and 128GB of memory.
To ensure fairness, we matched the number of cores utilized in both instances to the number of threads used in \sys, choosing between 1 or 4 threads.

Our evaluation covers two network conditions:
(i) Local Area Network (LAN): bandwidth of 1Gbps and an approximate latency (ping time) of 2ms;
(ii) Wide Area Network (WAN): we use the $\mathsf{tc}$ command to artificially constrain the bandwidth to 100Mbps and set the ping time to 100ms.


\subsection{Evaluation on \sys}
\label{sec:eva_on_sys}

\subsubsection{Datasets}
The evaluation datasets are shown in Table~\ref{tab:eva_datasets}.

\noindent $\bullet$ \textbf{iDash500K/1M}~\cite{idash2022}: 
iDash, a competition dedicated to addressing privacy and security challenges in genome data analysis, introduced a task in 2022 to tackle secure record linkage for electronic health records. Each record encompassed 9 attributes: Social Security Number (SSN), gender, first name, last name, birth date, phone number, address, state, and email. The SSN served as a deterministic key, whereas the other attributes played a supplementary role in record matching due to the potential missing in the SSN. 
All attributes were hashed using SHA-256, therefore, exact matching was required for comparison. 
iDash supplied two datasets, each consisting of two subsets, A and B, with each subset containing 500K records. We merged the two 500K datasets to create a larger dataset of 1M records in each subset.
In evaluation, we adopt a logistic model in the~\emph{Score Function} module and evaluate the effectiveness (accuracy) with 5-fold cross validation.
When evaluating the efficiency (runtime), we use all records.


\noindent $\bullet$ \textbf{DBLP/ACM}~\cite{kopcke2010evaluation}: 
This dataset, sourced from DBLP and ACM, contains bibliographic data with 2,615 and 2,293 records, respectively~\cite{github_cfb}. It includes 4 attributes: title, authors, venue, and year. Variations in text sequencing and abbreviations may occur within a single record across the two datasets, necessitating the use of approximate matching.
In evaluation, we adopt the simple linear model in \emph{Score Function} module. We set $B=25, 18$ for schema-aware and schema-agnostic modes, respectively, in approximate matching, implying that the number of attributes expands to 50 and 18. 



\noindent $\bullet$ \textbf{BNB/TPL}:
The British National Bibliography (BNB)~\cite{bnb}, includes a catalog of books published in the UK or the Republic of Ireland, and the Toronto Public Library (TPL)~\cite{tpl}, a collection of published books in the Toronto Public Library, serve as the datasets for this study.
We have extracted three fields—ISBN, author, and title—from each dataset, eliminating any record missing all three. 
To establish the ground truth, we follow the methodology outlined in SFour~\cite{khurram2020sfour}. 
ISBN can serve as a deterministic key, however, it is frequently missing.
For records without ISBNs, we adopt the Levenshtein distance for the author and title fields, setting a similarity threshold of 80\%. 
As a result, around 17.5\% of BNB records and 39.5\% of TPL records are being linked. Notably, around 57.1\% of these linked pairs can be matched through exact matching based on ISBN, consistent with findings in~\cite{khurram2020sfour}.
Our evaluation employs mixed matching modes: exact matching for the ISBN, and fuzzy matching for the author and title fields.
In evaluation, we adopt the simple linear model in \emph{Score Function} module. We set $B=8$ for both schema-aware and schema-agnostic modes, extending the number of attributes to 17 and 9, respectively. 

\begin{figure}[t]
\centering
 \includegraphics[width=0.95\linewidth]{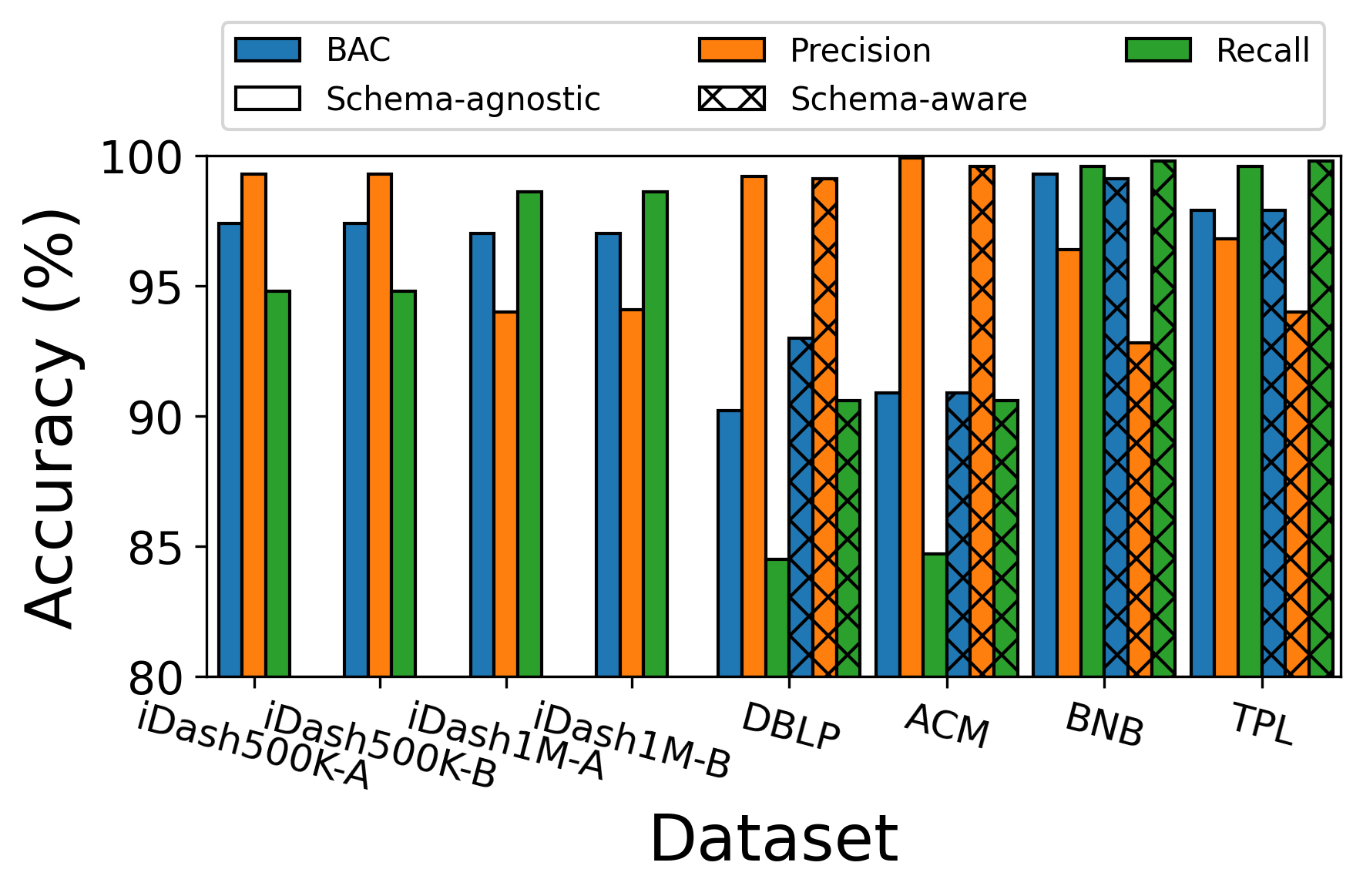}
\caption{Accuracy across various datasets, with $P_0$ utilizing one set of data from a given pair.}
\label{fig:eva_sys_accuracy}
\vspace{-5pt}
\end{figure}

\begin{table}[t!]
    \centering
    \caption{Runtime of \sys\ with 1/4 threads in seconds.}
    \begin{tabular}{|c|c|c|c|c|c|}
        \hline 
         &  \multicolumn{2}{c|}{Schema-aware} & \multicolumn{2}{c|}{Schema-agnostic} \\
        \hline
         dataset & LAN &WAN  &  LAN & WAN \\
        \hline
        iDash500K & 610/155 & 675/248 & / & /  \\
        iDash1M   & 1274/386 & 1331/492 & / & / \\
        DBLP   & 21/8 & 217/66 & 8/3 & 85/27  \\
        ACM  & 21/7 & 213/62 & 8/3 & 83/26  \\
        BNB & 5032/1710 & 5402/2025 & 2664/846 & 2861/990  \\
        TPL & 20165/9328 & 21317/9934 & 10558/4559 & 11281/4855  \\
        \hline
    \end{tabular}
    \label{tab:eva_sys_time}
\end{table}


\subsubsection{Evaluation Metrics}
To conduct a detailed analysis of \sys, we temporarily shift our system's objective to identifying whether a record is linked, rather than determining the collaboration value. This can be achieved by recovering $\langle \mathbf{d} \rangle^B$ 
within the \emph{Score Function} module. 
For effectiveness evaluation, we use metrics: Balanced Accuracy (BAC) $= 0.5 \times (\frac{TP}{TP+FN} + \frac{TN}{TN+FP}$); Precision $= \frac{TP}{TP+FP}$; Recall $= \frac{TP}{TP+FN}$ where $TP$, $TN$, $FP$, and $FN$ stand for true positives, true negatives, false positives, and false negatives, respectively.
In terms of efficiency, we measure its runtime, including both online and offline phases, under LAN and WAN.

\subsubsection{Results}
\label{sec:eva_sys_result}
Presented in Fig.~\ref{fig:eva_sys_accuracy}, the results are based on the dataset used by $P_0$ and alternating between the two datasets in a pair. 
This demonstrates that
\sys\ can achieve high accuracy when configured with appropriate parameters.

The runtimes of \sys\ with single-threaded and four-threaded configurations are displayed in~\Cref{tab:eva_sys_time}. These results confirm \sys's capability to efficiently process records at the scale of millions. Additionally, as the record count increases, the relative time difference between LAN and WAN scenarios decreases, indicating that the communication cost is no longer the primary bottleneck in the OFA protocol.

\subsubsection{Compare to Related Works}
\label{sec:eva_sys_compare_related_works}
In this section, we limit the number of threads to one and set the network condition to LAN for a fair comparison because some selected works either do not support multi-threading or have not implemented the network communication module in their open-source codes. Thus, we estimate their performance based on both their codes (if they exist) and the results reported in their papers.


\begin{table}[t!]
    \centering
    \caption{Estimated runtime (in hours) of PPRL and PPRS. }
    \begin{tabular}{|c|c|c|c|c|}
        \hline 
        Dataset & iDash500K & iDash1M & DBLP/ACM & BNB/TPL \\
        \hline
        \textbf{Ours} & $\mathbf{0.2}$  & $\mathbf{0.4}$  & $\mathbf{0.002}$/$\mathbf{0.002}$ & $\mathbf{0.7}$ / $\mathbf{2.9}$ \\        
        \cite{armknecht2023strengthening} & 4166.7 & 16666.7 & 0.1 & 34707.0 \\ 
        \cite{wei2023cryptographically} & 33.8 & 71.1 & 0.5 & 481.3 \\
        \cite{he2017composing}($\epsilon=0.1$) & 1456.8 & 3943.7  & 559.3& 436879.4 \\ 
        \cite{he2017composing}($\epsilon=0.4$) & 154.6 & 323.6 & 39.2 & 31034.8 \\ 
        \cite{he2017composing}($\epsilon=1.6$) & 36.0 & 75.6 & 3.6 & 3021.3 \\
        \hline
    \end{tabular}
    \label{tab:compare_pprl}
\end{table}

\begin{table}[t!]
    \centering
    \caption{The maximum $\alpha$  that \emph{Screening-the-Linkage} framework has an advantage over the pure PPRL framework. }
    \begin{tabular}{|c|c|c|c|c|}
        \hline 
        Dataset & iDash500K & iDash1M & DBLP/ACM & BNB/TPL \\
        \hline    
        \cite{armknecht2023strengthening} & 1.0 & 1.0 & 0.980/0.980 & 1.0/1.0 \\ 
        \cite{wei2023cryptographically} & 0.995 & 0.995 & 0.995/0.995 & 0.998/0.994 \\
        \cite{he2017composing}($\epsilon=0.1$) & 1.0 & 1.0  & 1.0/1.0& 1.0/1.0 \\ 
        \cite{he2017composing}($\epsilon=0.4$) &  0.999 & 0.999  & 1.0/1.0& 1.0/1.0 \\ 
        \cite{he2017composing}($\epsilon=1.6$) & 0.995 & 0.995 & 0.999/0.999 & 1.0/0.999 \\
        \hline
    \end{tabular}
    \label{tab:compare_framework}
\end{table}


We identify that SFour~\cite{khurram2020sfour} is PPRS
due to it outputs whether the record is linked instead of the linked pair. 
Compared to it, \sys\ demonstrates superior performance in terms of both effectiveness and efficiency. While SFour achieves slightly less than 90\% accuracy on the BNB/TPL dataset, \sys\ consistently maintains an accuracy rate exceeding 95\% across various settings. Furthermore, SFour takes approximately 385 minutes (77 minutes online and 308 minutes offline) to process 4096 records in a LAN environment. In contrast, \sys\ can efficiently handle up to 3.5 million records, as shown with the TPL dataset, in roughly 336 minutes under the same network conditions.


We also compare with the SOTAs on PPRL, BF-based (Bloom filter) PPRL~\cite{armknecht2023strengthening} and MPC-based PPRL~\cite{he2017composing,wei2023cryptographically}.

\noindent $\bullet$ BF-based PPRL: It comprises three components: diffusion, encoding, and linking. The diffusion and encoding processes are performed once per dataset, while the linking process is executed for every pair of records, resulting in time complexity of $\mathcal{O}(n^2)$. According to the reported results and our measurements using the available source code~\cite{github_spprld}, linking two records is the most consuming part, which takes around $6\times 10^{-5}$ ms. 


\noindent $\bullet$ MPC-based PPRL: This approach involves blocking records into bins using LSH, adding dummy elements in bins for privacy, and finally comparing the records in the respective bins from both parties. 
The difference lies: (i) He~\etal leverages differential privacy padding, while Wei~\etal~\cite{wei2023cryptographically} padding in a strategy of frequency smoothing; (ii) He~\etal adopts one-to-one secure comparison protocol while Wei~\etal adopts private join protocol (many-to-many secure comparison protocol).
Nevertheless, the cost is primarily determined by the number of secure pairwise comparisons. 
Our implementation of the comparison protocol takes around 0.13ms per comparison. 
Based on the analysis in~\cite{wei2023cryptographically}, the expected number of comparisons is approximately $B\ell(n/\ell+\eta)^2$ for both works (definitions are in Section~\ref{sec:problem_pprl}).
$\eta$ is different for two works, it is related to the different privacy parameter $\epsilon$ in He~\etal, while is $n/2\ell$ in Wei~\etal.
The specific parameters and settings are detailed in the original text.
For the iDash500K, iDash1M, DBLP/ACM, and BNB/TPL datasets, we tailored the parameters, $B=32$, to ensure the accuracy rate of their approaches is greater than 90\%.


For iDash datasets, we estimate the runtime of all works under the schema-aware setting, and for other datasets, we adopt the schema-agnostic setting.
The results are presented in Table~\ref{tab:compare_pprl}.
In the largest dataset BNB/TPL, the performance of \sys\ is exceptionally notable. Compared to the BF-based and MPC-based PPRL, \sys\ achieves a speedup factor of at least $11968 \times$ and $165\times$, respectively. 
These comparisons underscore the significant efficiency gains made possible by \sys, particularly when handling extensive datasets.

\subsection{Evaluation on Screening-then-Linkage}
\label{sec:eva_framework}
In this section, we compare the \emph{Screening-then-Linkage} framework with the framework solely based on PPRL.
We also use the same datasets, parameters, and environment for both PPRL and \sys\ as Section~\ref{sec:eva_sys_compare_related_works}.
In our setup, $P_0$ utilizes one dataset from a pair, while $P_1, ..., P_N$ employs the other dataset in that pair. We introduce $\alpha$, representing the proportion of parties possessing a high collaboration value ($c > \text{threshold}$).
To demonstrate the efficiency advantages of the \emph{Screening-then-Linkage} framework over pure PPRL framework, we compute the ratio of runtime between them, denoted as $\gamma$. Specifically, the runtime of the \emph{Screening-then-Linkage} framework is $T_1=N \cdot t_{\mbox{PPRS}} + \alpha \cdot N \cdot t_{\mbox{PPRL}}$.
Correspondingly, the running time of the pure PPRL framework is $T_2 = N \cdot t_{\mbox{PPRL}}$.
Thus, $\gamma = T_1/T_2 = t_{\mbox{PPRS}} / t_{\mbox{PPRL}} + \alpha$.

We determine the maximum value of  $\alpha$ for which the \emph{Screening-then-Linkage} has an advantage over the pure PPRL framework by solving $\gamma \geq 1$.
The results, shown in Table~\ref{tab:compare_framework}, indicate that the maximum $\alpha$ is larger than $0.98$ in all cases.
 This suggests that when nearly all parties are high-value participants for $P_0$, pure PPRL is more advantageous; otherwise \emph{Screening-then-Linkage} is better.

\begin{figure*}[t]
    \begin{minipage}[t]{.243\linewidth}
        \centering
        \includegraphics[width=\linewidth]{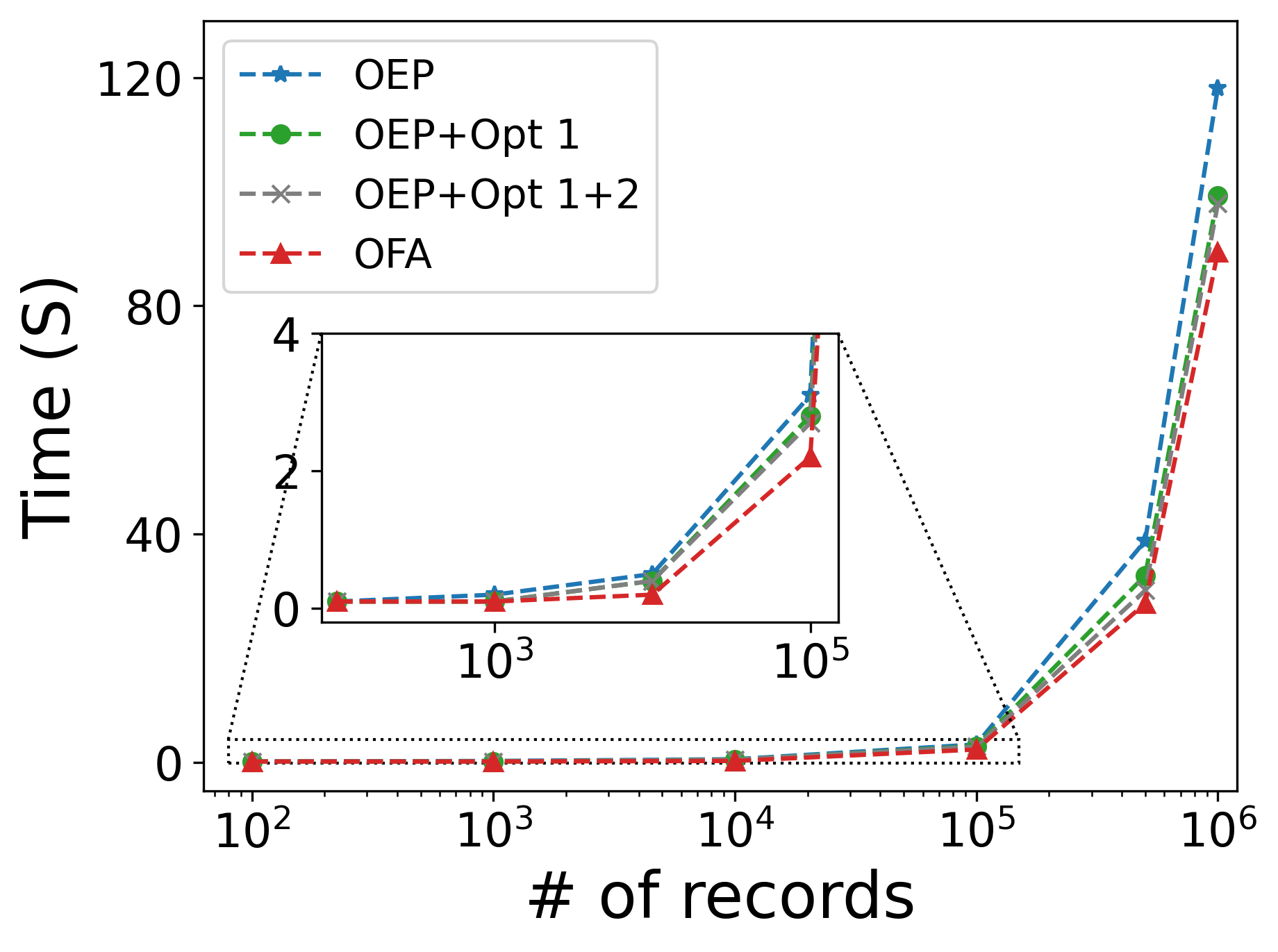}
        \subcaption{LAN, 1 thread}
        \label{fig:eva_OFA_lan_time_1}
    \end{minipage}
    \hfill
    \begin{minipage}[t]{0.243\linewidth}
        \centering
        \includegraphics[width=\linewidth]{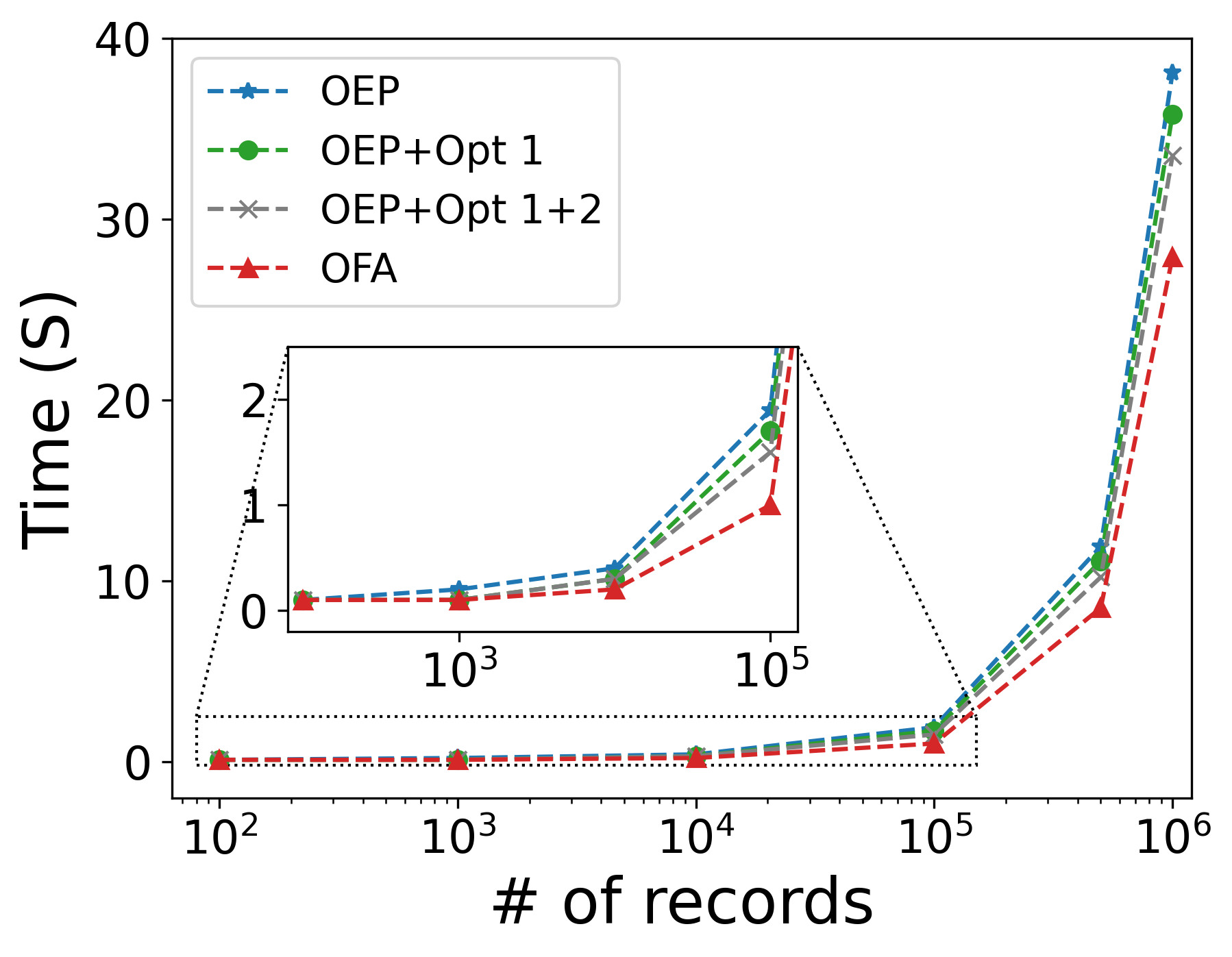}
        \subcaption{LAN, 4 threads}
        \label{fig:eav_OFA_lan_time_4}
    \end{minipage}  
    \hfill
    \begin{minipage}[t]{0.243\linewidth}
        \centering
        \includegraphics[width=\linewidth]{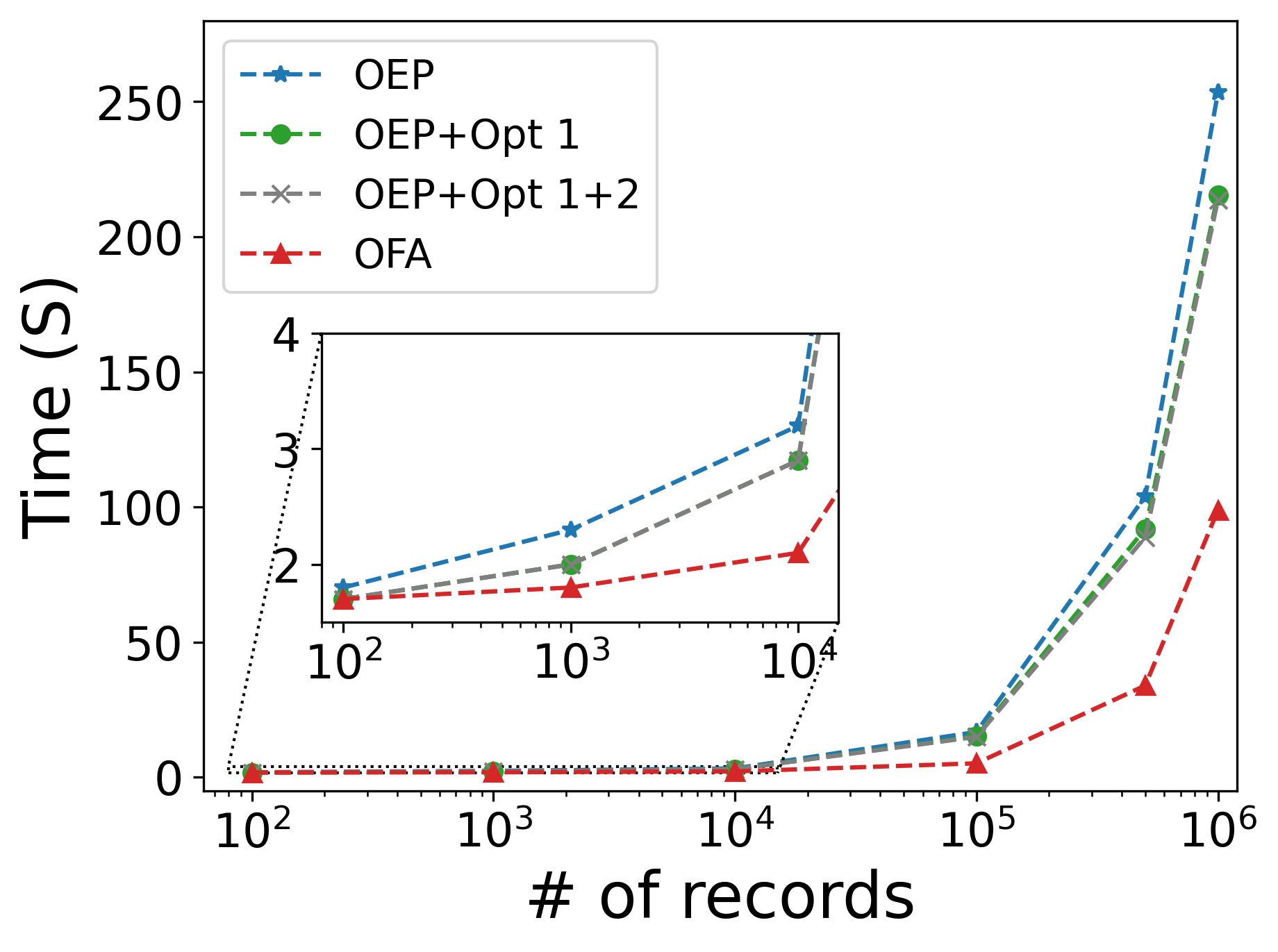}
        \subcaption{WAN, 1 thread}
        \label{fig:eav_OFA_wan_time_1}
    \end{minipage}  
    \hfill
    \begin{minipage}[t]{0.243\linewidth}
        \centering
        \includegraphics[width=\linewidth]{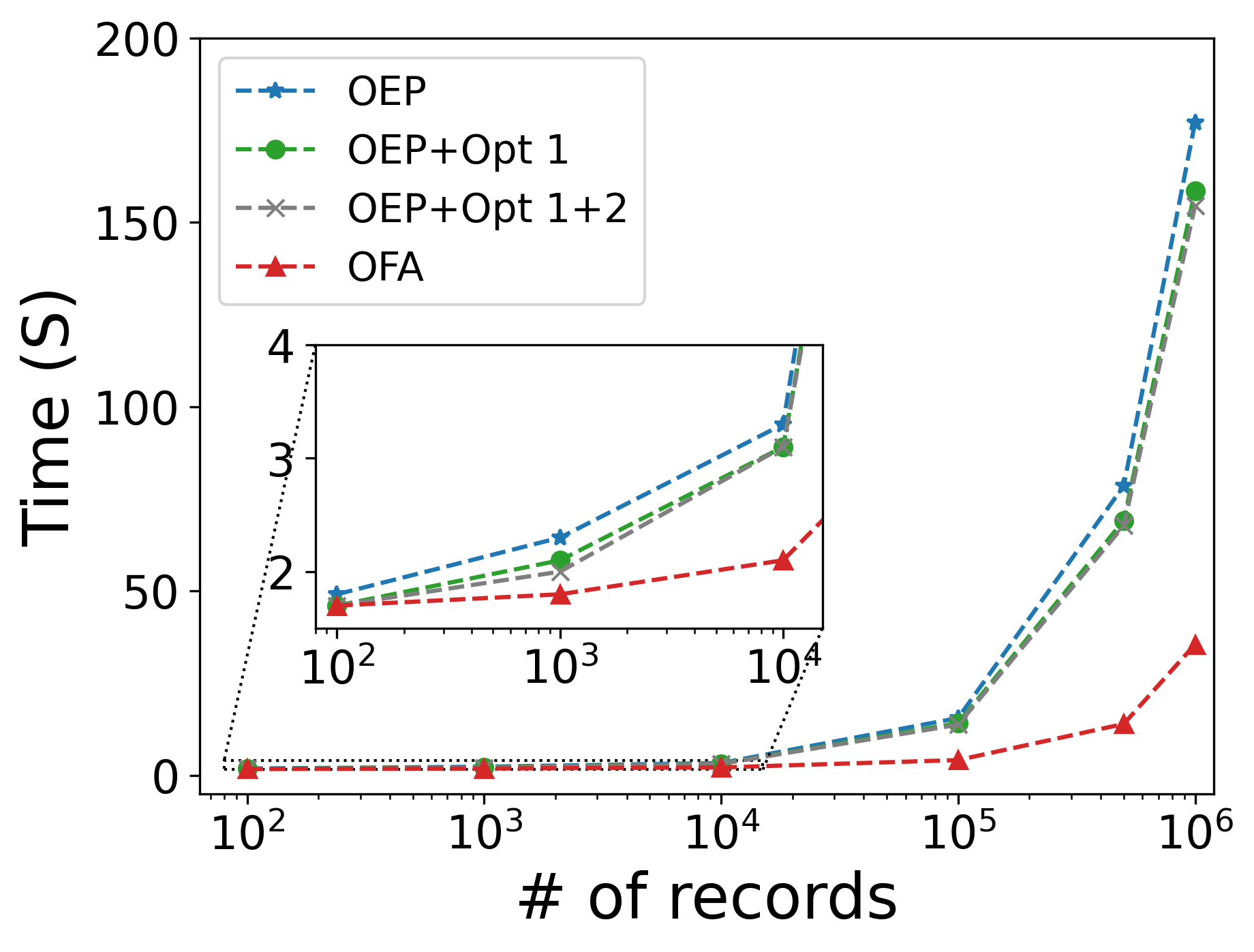}
        \subcaption{WAN, 4 threads}
        \label{fig:eav_OFA_wan_time_4}
    \end{minipage}  
    \caption{Running time (in seconds) of OEP, optimization 1, optimization 1+2, and OFA with 1 or 4 threads.}
    \label{fig:eav_OFA_time}
\end{figure*}


\begin{table}[t]
    \centering
    \caption{Communication cost (MB) of OEP, OEP with optimization 1, OEP with optimization 1+2 and OFA}
    \begin{tabular}{|c|c|c|c|c|c|c|}
        \hline 
        \# of records & 100 & 1K & 10K & 100K & 500K & 1M  \\ 
        \hline
        OEP & 0.3 & 1.1 & 11.6 & 136.8 & 799.7 & 1676.9 \\
        \hline
        OEP+Opt 1 & 0.3 & 1.0 & 10.3 & 121.4 & 695.2 & 1459.7 \\
        \hline
        OEP+Opt 1+2 & 0.3 & 1.0 & 10.0 & 118.4 & 680.0 & 1429.1 \\
        \hline
        OFA & 0.2 & 0.3 & 0.9 & 8.4 & 46.7 & 97.8 \\
        \hline
    \end{tabular}
    \label{tab:eva_OFA_communication}
\end{table}

\subsection{Component Evaluation: OFA protocol}
\label{sec:eva_ofa}
We evaluate the performance of the most critical module: OFA and compare it to the OEP with 2 instances, $P_0$ and $P_1$. 
A variety of datasets are synthesized,
each containing a single attribute of random boolean shares.
The number of records in these datasets varies from 100 to 1M.
Additionally, we generate random extended permutation $\pi$ for $P_0$. 

\Cref{fig:eav_OFA_time} presents the runtime results for OEP, optimizations and the OFA protocols. 
In a LAN setting, the OFA protocol demonstrates a slight performance advantage over the OEP. 
However, the distinction becomes more pronounced in a WAN environment, where the OFA significantly outperforms the OEP. 
This improvement in performance is attributed to the reduced communication cost associated with the OFA, as detailed in \Cref{tab:eva_OFA_communication}. 
The efficiency gains become increasingly evident as the dataset size expands. For a dataset comprising 1M records, the runtime of OFA is $4\times$ faster than that of the OEP with 4 threads, the communication cost of the OFA is a mere 6\% of the OEP's, highlighting the substantial efficiency improvements achieved by the OFA protocol.
Moreover, optimization 1 and 2 under LAN contribute almost equally as they do under WAN, and optimization 3 contributes more under WAN.
It indicates that optimization 3 reduces the communication cost significantly, and the other two optimizations reduce both communication and computation costs.

\section{Related Work}

There are considerable works dedicated to achieving PPRL.
All-Pairwise Comparisons~\cite{he2017composing} is a straightforward and accurate approach. It involves secure comparison between each pair of records in the Cartesian product of the two databases. However, due to its quadratic time complexity, it becomes inefficient for large-scale datasets. 
As a result, researchers have introduced various solutions to strike a balance between security, effectiveness, and efficiency. 
These approaches involve three primary techniques, including Differential Privacy (DP), Bloom Filter (BF), and Multi-party Computation (MPC).


MPC-based approaches, while provably secure and accurate, are generally inefficient. Some works~\cite{wong2013privacy, essex2019secure} rely on computationally inefficient homomorphic encryption, while others~\cite{chen2018perfectly} are based on communicationally inefficient garbled circuits. Among them, the PSI-based approach~\cite{adir2022privacy} achieves the best efficiency and can support more than 1M records. However, it can only handle the exact match in schema-agnositc scenario. 
Adir~\etal~\cite{adir2022privacy} tries to solve the approximate match by directly running PSI on the elements extracted from LSH on each record.
However, this leaks the frequency of each q-gram, and can inadvertently leak information about the attributes of the records. 
Moreover, it cannot support the schema-aware mode.
In our work, we prevent such information leakage through the use of circuit-PSI.
Wei~\etal~\cite{wei2023cryptographically} leverages LSH for blocking and private join protocol for matching. It achieves the best performance in the MPC-based approach but still needs 1 hour for 40K records.
DP-based works~\cite{inan2010private, cao2015hybrid, kuzu2013efficient} could improve efficiency. 
However, He~\etal~\cite{he2017composing} demonstrate that these approaches may leak aggregate properties of the input datasets to both parties and propose a more secure solution based on the combination of DP and MPC. 
Rao~\etal~\cite{rao2019hybrid} achieves the best performance in the DP setting but relies on a trusted third party. 
BF-based works~\cite{schnell2009privacy, boyd2015application} are often efficient. However, they are vulnerable to various attacks~\cite{kuzu2011constraint, niedermeyer2014cryptanalysis, vidanage2019efficient, vidanage2020graph}. 
Moreover, the solutions~\cite{christen2020linking, ranbaduge2020securing, franke2021evaluation} 
lack formal security analysis. 
A recent work~\cite{armknecht2023strengthening} enhances the security of BF by adding a diffusion layer, providing a rigorous security analysis. But the complexity is $\mathcal{O}(n^2)$ which is hard to scale.
For further details, we direct readers to refer~\cite{vatsalan2013taxonomy, vatsalan2017privacy, gkoulalas2021modern}.


Some recent works are highly relevant to PPRS. 
SFour~\etal~\cite{khurram2020sfour} proposes using MPC-based sorting with sliding window comparison to enable the output of which records are linked.
However, the cryptographic primitives used are time-consuming, resulting in the entire system taking more than 1 hour for only 4096 records.
Along a similar line of work, we systematically define the PPRS framework. Unlike their approaches, our solution achieves $\mathcal{O}(n\log n)$ complexity based on optimized lightweight cryptographic building blocks, thereby supporting million-level records.


\section{Conclusion}
\label{sec:conclusion}
In this work, we point out the need for efficient collaboration value identification in the data federation market and introduce the \emph{Screening-then-Linkage} framework to address this. The superior performance of \sys\ depends on circuit-PSI and our proposed OFA protocol. With a thorough security analysis, \sys\ is demonstrated to leak no information beyond the data collaboration value. Our evaluations show that \sys\ only needs about 20 minutes for a 1M dataset, which is $165\times$ faster than the SOTA PPRL. In our future work, we aim to address more practical considerations related to the deployment in data markets, including: 
supporting more advanced approximate matching methods with circuit-PSI instead of expanding the records; 
implementing more efficient PPRS based on other privacy-preserving techniques such as DP and BF; 
analyzing the appropriate setting of parameters to avoid rare cases where the attribute of $P_0$ is linked, but the linked attributes of $P_1$ are from different records, leading to the failure of record screening;
integrating data pricing methods with our \emph{Screening-then-Linkage} framework and \etc

\newpage
\clearpage
\bibliographystyle{IEEEtran}
\bibliography{reference}

\end{document}